\documentclass[10pt]{article}
\makeatletter\def\input@path{{./}{arxiv/source/}}\makeatother
\usepackage{amsmath, amssymb, amsthm, amsfonts, latexsym}
\usepackage{mathtools}
\usepackage{enumerate}
\usepackage[sc]{mathpazo}
\usepackage[scaled=0.92]{helvet}
\usepackage[T1]{fontenc}
\usepackage{microtype}
\usepackage{graphicx}
\usepackage[letterpaper,margin=0.75in,columnsep=0.25in]{geometry}
\makeatletter
\renewcommand{\fnum@figure}{Figure~\thefigure}
\renewcommand{\fnum@table}{\textbf{Table \thetable}}
\long\def\@makecaption#1#2{\vskip\abovecaptionskip{\small\def\textbf##1{{\sffamily\bfseries ##1}}\textbf{#1}\enspace #2\par}\vskip\belowcaptionskip}
\makeatother
\usepackage[square,comma,numbers,sort&compress]{natbib}
\usepackage{url}
\allowdisplaybreaks[4]

\newtheorem{theorem}{Theorem}
\newcommand{\eps}{\varepsilon}

\def\scititle{Multi-Agent Learning of Reactive Strategies in Iterated Games}
\makeatletter
\renewcommand\section{\@startsection{section}{1}{\z@}{-3ex \@plus -1ex \@minus -.2ex}{1.2ex \@plus .2ex}{\sffamily\bfseries\Large}}
\renewcommand\subsection{\@startsection{subsection}{2}{\z@}{-2.5ex \@plus -1ex \@minus -.2ex}{1ex \@plus .2ex}{\sffamily\bfseries\large}}
\renewcommand\subsubsection{\@startsection{subsubsection}{3}{\z@}{-2ex \@plus -1ex \@minus -.2ex}{0.6ex \@plus .2ex}{\sffamily\bfseries\itshape\normalsize}}
\renewcommand\paragraph{\@startsection{paragraph}{4}{\z@}{1.5ex \@plus .5ex \@minus .2ex}{-1em}{\sffamily\bfseries\small}}
\renewcommand{\thesection}{\Roman{section}}
\renewcommand{\thesubsection}{\Alph{subsection}}
\renewcommand{\p@subsection}{\thesection\,}
\makeatother
\usepackage{xcolor}
\usepackage[colorlinks=true,allcolors={rgb:red,1;green,3;blue,8}]{hyperref}

\begin{document}
\twocolumn[{%
\begin{flushright}
{\sffamily\bfseries\huge \scititle\par}
\vspace{12pt}
{\sffamily\bfseries\small James P.~Wolfe$^{1}$, Joshua B.~Plotkin$^{2,\ast}$\par}
\vspace{3pt}
{\sffamily\scriptsize $^{1}$Graduate Group in Applied Mathematics and Computational Science, University of Pennsylvania, Philadelphia, PA 19104, USA, $^{2}$Department of Biology, University of Pennsylvania, Philadelphia, PA 19104, USA\par
$^\ast$Corresponding author. Email: jplotkin@sas.upenn.edu\par
\vspace{3pt}
October 9, 2026\par}
\end{flushright}
\vspace{22pt}
{\sffamily\small\noindent\textbf{ABSTRACT}\enspace Learning agents, whether humans or algorithms, adapt to one another rather than to a fixed environment. Selfish learning need not produce a good outcome for the group, especially when learners face a social dilemma. The iterated prisoner's dilemma (IPD) is the classical example, but it is usually studied in populations of agents who imitate successful strategies. Here, we instead study selfish \textit{gradient ascent learners} on the two-dimensional space of \textit{reactive strategies}. A reactive strategy is described by two continuous traits: reciprocity, the chance of cooperating after the co-player cooperates, and forgiveness, the chance of cooperating after the co-player defects. Each learner updates its strategy in the direction that most rapidly improves its payoff against all other learners, by an amount governed by its \textit{learning rate}. We show that heterogeneity in learning rates can qualitatively change the long-term outcome: a population with uniform learning rates converges to always-defect (ALLD), whereas a population with variable learning rates approaches generous-tit-for-tat (GTFT), which produces nearly optimal collective payoffs. In simulations, every population passes through a transient state of minimal forgiveness, in which the distribution of reciprocity determines the outcome. Variable learning rates promote cooperation because slow learners preserve high reciprocity until the population reaches this state. We prove sufficient conditions on the distribution of reciprocity for convergence to ALLD or to tit-for-tat (TFT). Variable learning rates also allow cooperation to recover from perturbation, which is provably impossible under uniform learning. These results suggest that purely selfish and myopic learning can produce persistent cooperation provided the population has intrinsic variability in learning rates.\par}
\vspace{24pt}
}]

\section{Introduction}

Cooperation is costly to the individual who provides it, and yet it underlies social life at every scale of biological organization. In a single encounter between self-interested players, the prisoner's dilemma admits only one Nash equilibrium, mutual defection, and cooperation cannot persist in a well-mixed population. Repeated interaction changes this. When the same players meet many times, each can condition its behavior on the history of play, so that the prospect of future reward or retaliation, the ``shadow of the future,'' makes reciprocity possible~\cite{trivers1971-EvolutionReciprocalAltruisma,axelrod1981-EvolutionCooperation}. Reciprocal strategies such as tit-for-tat and generous tit-for-tat reward cooperation and forgive or punish defection, and they can stabilize cooperation that would be unstable in a one-shot game~\cite{nowak1990-EvolutionStochasticStrategies,nowak1992-TitTatHeterogeneous,nowak2006-FiveRulesEvolution}. The iterated prisoner's dilemma (IPD) is the canonical framework for studying how reciprocity allows cooperation among self-interested individuals to arise and persist~\cite{axelrod1981-EvolutionCooperation,nowak2006-EvolutionaryDynamicsExploring}.

The same repeated structure that makes reciprocity possible also makes the IPD strategically rich. A strategy may condition each move on the entire history of play, so the space of possible strategies is enormous. The Folk Theorem of repeated game theory makes the size of the resulting equilibrium set precise: once the horizon is long enough, nearly any individually rational payoff can be supported at a Nash equilibrium, and the set of equilibria expands as the game lengthens~\cite{friedman1971-NoncooperativeEquilibriumSupergames,benoit1985-FinitelyRepeatedGames,fudenberg1986-FolkTheoremRepeated}. This multiplicity is what makes cooperation possible, but it also leaves the outcome undetermined, since equilibrium analysis alone cannot say which of the many equilibria a population will reach. A strategy space of this size is a challenge for any learning rule, and for this reason the IPD has become a natural testbed for multi-agent learning~\cite{sandholm1996-MultiagentReinforcementLearning,bloembergen2015-EvolutionaryDynamicsMultiAgent}. At the same time, the learning rules studied in this setting may capture something about how humans and other organisms adjust their behavior, not only how artificial agents might be designed.

Two research traditions address strategic change in this game. Evolutionary game theory models populations whose strategies spread by selection or imitation, whereas the theory of learning in games models individuals who adjust their own strategies in response to payoffs~\cite{fudenberg1998-TheoryLearningGames}. Here we study a population of selfish learners who update reactive IPD strategies by gradient ascent on their payoffs against the population as a whole. We show that a single, naturalistic form of individual variation, namely heterogeneity in learning rates, can determine whether such a population converges to unconditional defection or to stable and nearly socially optimal cooperation.

Evolutionary game theory provides a natural point of comparison for our study~\cite{smith1982-EvolutionTheoryGames}. In that framework, strategies change by payoff-biased imitation or differential reproduction rather than by gradient ascent~\cite{taylor1978-EvolutionaryStableStrategies,hofbauer1998-EvolutionaryGamesPopulation}, and the aim is not to design effective algorithms for multi-agent problems but to model how strategies spread and are inherited among humans or other organisms. The IPD, in particular, is used to model the evolution of altruistic behavior. Much of this research considers a discrete set of strategies, most notably evolutionary versions of Axelrod's tournament~\cite{axelrod1981-EvolutionCooperation}, but several studies treat continuous strategy spaces. Nowak \& Sigmund~\cite{nowak1992-TitTatHeterogeneous} studied the space of reactive memory-one strategies, in which a player's policy depends on the co-player's previous move. In their simulations of payoff-biased imitation, populations usually evolve toward unconditional defection; occasionally, the reciprocal strategy tit-for-tat permits the rise of the more forgiving strategy generous tit-for-tat, which resists invasion by defectors.
Imhof \& Nowak~\cite{imhof2010-StochasticEvolutionaryDynamics} extended this analysis by adding mutation. They found the same role for tit-for-tat, but with mutation the population does not remain at generous tit-for-tat: it cycles between cooperative and uncooperative phases indefinitely.

We instead consider a gradient-ascent learning process in a population of agents engaged in pairwise iterated games. Each individual plays the rest of the population and updates its strategy locally, in the direction that most increases its payoff. Gradient ascent assumes that players have enough cognition to compute their gradient, but that they are myopic, making only local changes to their strategies. Gradient ascent therefore sits between two extremes: classical economic analysis, which assumes perfect rationality and immediate convergence to a Nash equilibrium, and evolutionary game theory, which assumes payoff-biased imitation with no cognition at all. Neither extreme predicts an outcome here. The Folk Theorem~\cite{fudenberg1986-FolkTheoremRepeated} guarantees a continuum of Nash equilibria for infinitely repeated games but gives no guidance as to which one is chosen, and payoff-biased imitation can cycle indefinitely without settling on any stable equilibrium~\cite{imhof2010-StochasticEvolutionaryDynamics,nowak2006-EvolutionaryDynamicsExploring,imhof2005-EvolutionaryCyclesCooperation}.

Prior work on learning dynamics in repeated games falls into two broad categories: (i) theoretical analysis of convergence for RL or gradient dynamics, often in two-player, two-action games; and (ii) empirical studies evaluating a wider range of learning algorithms across broader classes of games. Empirical studies of this second kind are more relevant to us. Several empirical studies show that learning agents, including but not limited to those using gradient ascent, frequently fail to solve the social dilemma of the IPD, tending to reach a non-cooperative resolution~\cite{sandholm1996-MultiagentReinforcementLearning,macy2002-LearningDynamicsSocial,tuyls2006-EvolutionaryDynamicalAnalysis,bloembergen2015-EvolutionaryDynamicsMultiAgent}. These studies, however, often do not examine the dynamics of larger strategy spaces, considering only the one-dimensional ``probability of cooperating.'' In this one-dimensional case, theoretical studies show that under certain gradient ascent algorithms the strategies of both players converge on ``always defect,'' the only Nash equilibrium, or else their average payoffs converge on the payoffs at that equilibrium~\cite{singh2000-NashConvergenceGradient,bowling2002-MultiagentLearningUsing}. One study avoids the failure of gradient ascent in the IPD by allowing agents to include, with some weighting, not just their own payoff but also the average payoff of the group when calculating gradients~\cite{zhang2019-SAIGAMultiagentReinforcement}. In a larger strategy space, Foerster et al.~\cite{foerster2018-LearningOpponentLearningAwareness} demonstrated the difficulty of reaching cooperation without foresight. Studies of learning in the IPD also rarely consider the dynamics of \textit{many} agents updating their strategies simultaneously. McAvoy et al.~\cite{mcavoy2022-SelfishOptimizationCollective} show that while two agents updating their strategies against one another by gradient ascent typically fail to achieve socially optimal cooperation, a large population of learners well-mixed in their interactions, who encounter any one individual rarely, can succeed. But in that model the interactions, and updates, are random and occur one at a time.

Like Imhof \& Nowak, we study a large population on the space of reactive memory-one strategies, but our agents update by gradient ascent rather than imitation. As in the evolutionary game theory literature, our aim is not to design a more effective algorithm for multi-agent learning but to model how individuals adjust their strategies under a different set of assumptions. And in contrast to McAvoy et al.~\cite{mcavoy2022-SelfishOptimizationCollective}, where agents meet one at a time, we study learning that occurs \textit{simultaneously} across the population, so that each agent's update reflects the strategies of all others at once.

We are most interested in how individual variation in one parameter can affect population-level outcomes. There is reason to believe individual variation in microparameters might produce globally different outcomes, even though agents' updating rules are otherwise identical. Barfuss \& Meylahn~\cite{barfuss2023-IntrinsicFluctuationsReinforcement} show that individual stochastic fluctuations among agents undergoing reinforcement learning on the space of memory-one strategies in the IPD increased overall cooperation levels. Meng et al.~\cite{meng2024-DynamicsCollectiveCooperation} show that heterogeneous networks, when coupled with heterogeneous strategy updating rates, are more conducive to cooperation. We endow our agents with minimal knowledge: each agent knows only which direction of strategic change most rapidly improves its payoff. Beyond that direction, each agent has a single intrinsic parameter, its learning rate, that governs how strongly it acts on that information. We show that variation in the learning rate across the population, with no other source of noise or foresight, produces a population of seemingly stable cooperation where otherwise there would not be.

\section{Model}

\subsection{The iterated prisoner's dilemma with reactive strategies}

We consider a simplified iterated prisoner's dilemma (IPD) where, in each round, a pair of agents simultaneously choose whether to cooperate, paying a cost $c$ to confer a benefit $b > c$ to the other agent, or to defect, paying no cost and conferring no benefit~\cite{axelrod1981-EvolutionCooperation,sigmund2010-CalculusSelfishness,hilbe2018-PartnersRivalsDirect}. Strategies for the infinitely-repeated game can be arbitrarily complex, prescribing each round what action to take conditioned on the history of all prior rounds. We focus on a class of so-called ``reactive,'' memory-one strategies~\cite{nowak1990-EvolutionStochasticStrategies}, which are parameterized by two probabilities, $(p,q) \in [0,1]^2$, that specify the chance a player will cooperate each round in response to the action taken by the co-player in the preceding round. The parameter $p$ is the probability a player cooperates given their opponent cooperated in the previous round, i.e. their \textit{reciprocity}. The parameter $q$ is the probability a player cooperates given their opponent defected in the previous round, i.e. their \textit{forgiveness}. The space of reactive strategies includes many well-known examples such as always cooperate $(1,1)$, always defect $(0,0)$, tit-for-tat ($1,0$), and generous tit-for-tat ($1,1-\delta)$, which always reciprocates cooperation but forgives a defection with probability $1-\delta$~\cite{nowak1992-TitTatHeterogeneous}.

To ensure that the infinitely repeated game has a unique stationary distribution, and that long-term payoffs to both players do not therefore depend on the first round of play, we restrict ourselves to a compact subset of reactive strategies $(p,q)\in[\eps, 1-\eps]^2$, for small $0<\eps<1/2$. This strategy space can also be seen as reactive strategies in the context of a small rate of implementation error, $\eps$. The long-run average payoff per round to a strategy $\mathbf{x}_1 = (p_1, q_1)$ playing against $\mathbf{x}_2 = (p_2, q_2)$ is then given by
\[
\pi(\mathbf{x}_1, \mathbf{x}_2) = \frac{b(q_1 r_2 + q_2) - c(q_2 r_1 + q_1)}{1 - r_1 r_2},
\]
where $r_1 = p_1 - q_1$ and $r_2 = p_2 - q_2$ are the two players' \emph{degrees of reciprocity}, the amount by which each is more likely to cooperate after the co-player's cooperation than after a defection. The degree of reciprocity measures how strongly a player conditions its behavior on the co-player's last move, and it separates conditional from unconditional cooperation: tit-for-tat is maximally reciprocal, with $r = 1$, whereas always-cooperate and always-defect share $r = 0$, cooperating with the same probability whatever the co-player did.

\subsection{Strategy evolution by gradient ascent in a population of players}\label{sec:gradient-ascent}

We start by considering a well-mixed population of agents with fixed finite size, $N$. In successive discrete learning steps, agents play all pairwise games (including self-play) and receive their average payoff across those games, and then they each simultaneously update their strategy in the direction of their local gradient. Considering at first the simple case of just two agents playing $\mathbf{x}_1 = (p_1, q_1)$  against $\mathbf{x}_2 = (p_2, q_2)$, the gradient of agent 1's payoff with respect to its strategy, $\nabla_{\mathbf{x}_1}\pi = (\partial \pi/\partial p_1,\, \partial\pi/\partial q_1)$, has components
\begin{align}
\frac{\partial \pi}{\partial p_1} &= \frac{(b r_2 - c)(q_1 r_2 + q_2)}{(1 - r_1 r_2)^2}, \\
\frac{\partial \pi}{\partial q_1} &= \frac{(b r_2 - c)(1 - q_2 - p_1 r_2)}{(1 - r_1 r_2)^2}.
\end{align}
Both components of this gradient share the factor $b r_2 - c$, which has a simple interpretation. Agent 1's long-run payoff is $\pi = b x_2 - c x_1$, where $x_1$ and $x_2$ are the long-run probabilities that agents 1 and 2 cooperate in a round. Agent 2 cooperates with probability $q_2$ after a defection and with a probability larger by its degree of reciprocity, $r_2$, after a cooperation, so that $x_2 = q_2 + r_2 x_1$. Raising either of its cooperation probabilities, $p_1$ or $q_1$, raises agent 1's long-run rate of cooperation, $x_1$, and it changes agent 2's long-run rate of cooperation by $r_2$ times as much. Each unit of additional cooperation therefore costs agent 1 the amount $c$ and returns the benefit $b r_2$. Agent 1's payoff consequently increases in both of its strategy parameters when its co-player's degree of reciprocity exceeds $c/b$, and it decreases in both when the co-player's degree of reciprocity is below $c/b$. The threshold $c/b$ on an opponent's degree of reciprocity recurs throughout our analysis.

In a large population, however, each agent is assumed to play pairwise IPD games against all others in the population (including self-play) and update their strategy in the direction of maximum increase to their average payoff across all these pairwise games. The average payoff of strategy $\mathbf{x}_i$ in the population at time $t$ is
\[
\Pi^t(\mathbf{x}_i^t) = \frac{1}{N}\sum_{j=1}^N \pi(\mathbf{x}_i^t, \mathbf{x}_j^t),
\]
and so its gradient with respect to $x_i^t$ becomes
\[
\nabla_{\mathbf{x}_i^t} \, \Pi^t(\mathbf{x}_i^t) = \frac{1}{N} \sum_{j=1}^N \nabla_{\mathbf{x}_i^t} \, \pi(\mathbf{x}_i^t, \mathbf{x}_j^t).
\]
Note that we include self-play, though this does not meaningfully affect results.

In each discrete learning step, $t$, each individual updates their strategy according to a single step of gradient ascent:
\[\mathbf{x}_i^{t+1} = \operatorname{Proj}_{[\eps, 1-\eps]^2} \left(\mathbf{x}_i^t + \eta \nabla_{\mathbf{x}_i^t} \, \Pi^t(\mathbf{x}_i^t) \right), \]
where $\eta$ denotes the \emph{learning rate}, which scales the magnitude of an agent's strategic update. The operator $\operatorname{Proj}_{[\eps, 1-\eps]^2}$ ensures that the resulting strategy still lies within  $[\eps,1-\eps]^2$: this operator simply maps any strategy in $\mathbb{R}^2$ to the nearest point, by Euclidean norm, on the square $[\eps,1-\eps]^2$. In this case, this is equivalent to ``clipping'' each component to lie within $[\eps,1-\eps]$.

\begin{figure*}[t]
\centering
\includegraphics[width=\textwidth]{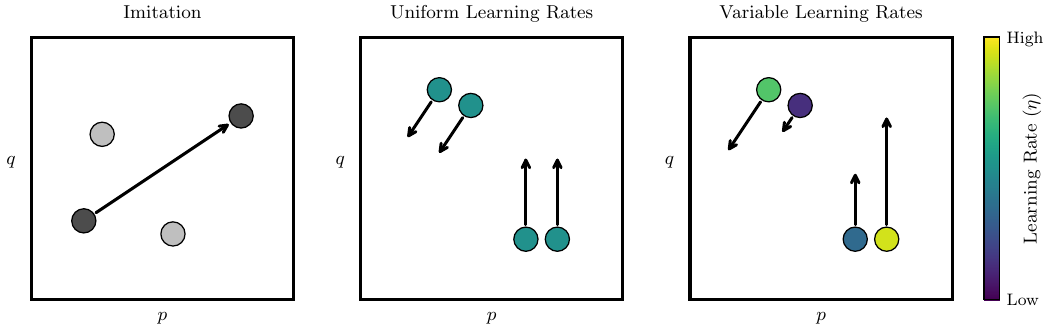}
\caption{\textbf{Agents starting from the same strategy move different distances when learning rates vary.} Comparison of different forms of strategy updating on the $(p,q)$ strategy space. Imitation involves one strategy copying another, typically in a manner such that more successful strategies are imitated more often. By contrast, agents who update via gradient ascent take into account other agents' strategies, but do not copy them; instead, they move in a direction which most rapidly improves their payoff. Having uniform learning rates means agents share a coefficient dictating how far along their gradients they move each step. Not all agents move the same distance, though agents at similar starting points should move similar distances. By contrast, with variable learning rates, each agent has their own coefficient which might differ from others, so agents starting from the same point with different learning rates move different distances.}
\label{fig:cartoon}
\end{figure*}

\subsection{Variation in rates of learning}

Now assume that each agent $i$ has its own intrinsic learning rate $\eta_i$ that does not change over time, but may differ among agents. Each agent $i$ at time $t$ is described by an ordered triple: $\mathbf{X}_i^t = (p_i^t, q_i^t, \eta_i)$, or $\mathbf{X}_i^t = (\mathbf{x}_i^t, \eta_i)$, where $\mathbf{x}_i^t = (p_i^t, q_i^t)$ and $\eta_i$ is unchanging in time. After playing all pairwise games and receiving average payoffs, agents simultaneously update their strategies by gradient ascent, using possibly different learning rates:
\[\mathbf{x}_i^{t+1} = \operatorname{Proj}_{[\eps, 1-\eps]^2} \left(\mathbf{x}_i^t + \eta_i \nabla_{\mathbf{x}_i^t} \, \Pi^t(\mathbf{x}_i^t) \right). \]
Two agents starting from the same strategy but differing in learning rate move different distances in a single update step (Fig.~\ref{fig:cartoon}, \emph{Variable Learning Rates}), whereas agents sharing a learning rate move similar distances (Fig.~\ref{fig:cartoon}, \emph{Uniform Learning Rates}), and agents who imitate copy one another's strategies outright (Fig.~\ref{fig:cartoon}, \emph{Imitation}), as in evolutionary game theory~\cite{smith1982-EvolutionTheoryGames,hofbauer1998-EvolutionaryGamesPopulation,nowak2006-EvolutionaryDynamicsExploring}.

We explore the discrete state and time process defined above, in which agents simultaneously perform gradient ascent updates, by direct simulation. The population is initialized with strategies drawn uniformly at random from $[\eps, 1-\eps]^2$. The deterministic process is simulated using either uniform learning rates, i.e. $\eta_i = \eta_j = \bar\eta$ for all $i, j$, or variable learning, where $\eta_i$ is drawn from $\mathrm{Uniform}[\bar\eta-a, \bar\eta + a].$ Unless otherwise noted, simulations use $c/b = 1/2$.

\subsection{Evolution of strategic density for infinite-population limit}

In the limit of large population size $N\to \infty$ we track the population's strategies not as individual agents, but as a generalized distribution (possibly a continuous density) $F^t(x)$ over the strategy space $[\eps, 1-\eps]^2$ at time $t$. Each strategy $x$ then moves along the average, over the whole population, of its payoff gradients against the strategies it plays. Even though the population's distribution may be continuous, we still consider discrete learning steps; taking a limit of small steps and rescaling time produces a continuous-time PDE, which we state in Sec.~S1.2 of the Supplemental Material. We also allow a distribution $G^t(\eta)$ of learning rates, independent of $x$ at $t = 0,$ and correlated with $x$ for $t>0$, since an agent's learning rate determines how far it moves at each step. This continuous process is the alternative ``Eulerian'' formulation of the individual-based ``Lagrangian'' formulation above.

To simulate this process, which admits continuous density functions, we approximate distributions with a fixed lattice over $[\eps, 1-\eps]^2$, each lattice point with an associated frequency (akin to a ``histogram approximation'' of a continuous distribution). Then, when mass moves between points it is redistributed via bilinear interpolation. These are standard methods for numerical solutions to PDEs, adapted to a discrete-time process. Details on the formulation and implementation of this continuous ``Eulerian'' model can be found in Sec.~S1 of the Supplemental Material.

\subsection{Initialization}

Global parameters $b = 1,c  = 1/2, \eps = 0.001, N = 2500$. Learning rates are either drawn $\eta_i \sim \mathrm{Uniform}[0.01, 0.39]$ (variable learning rates) or set $\eta_i = 0.2$ (uniform). The simulation for a variable learning rates population was terminated after 25,000 steps, and for a uniform learning rates population after 15,000 steps. Different $b/c$ ratios and population sizes $N$ are explored in Sec.~\ref{sec:robustness}.

\section{Results}

\begin{figure*}[t]
\centering
\includegraphics[width=\textwidth]{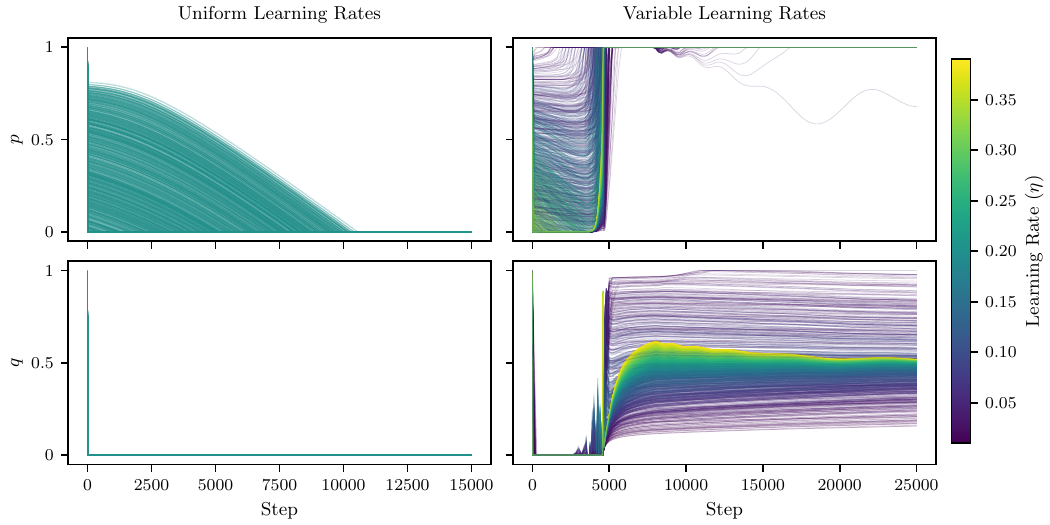}
\caption{\textbf{The population converges to generous tit-for-tat when learning rates vary, and to unconditional defection when they are uniform.} Trajectory plot of agents' reciprocity $p$ (top) and forgiveness $q$ (bottom) values over all of their learning steps, for agents with uniform learning rates (left) and variable learning rates (right).
}
\label{fig:distributions}
\end{figure*}

Individual differences in learning rates can qualitatively alter the course of strategic change in a population undergoing gradient ascent (Fig.~\ref{fig:distributions}). When agents all share the same learning rate ($\eta_i \equiv \eta$ for all $i$), the population converges to unconditional defection, eventually leaving every agent with payoff near zero. Convergence to defection is the outcome one would expect of selfish gradient learners in a social dilemma. But when agents vary in their learning rates, the population instead converges to a generous tit-for-tat strategy (GTFT), which sustains a high rate of cooperation and nearly socially optimal payoffs (Fig.~S1 of the Supplemental Material).

To understand how variation in learning rates produces these divergent outcomes, we consider a population of agents with strategies $\mathbf{x_i}=(p_i,q_i)$ initially drawn independently and uniformly from the space of reactive strategies $[\eps, 1-\eps]^2$. Individual learning rates $\eta_i$ are drawn from $\text{Uniform}[\bar\eta-a, \bar\eta + a]$, with $a = 0$ in the case of uniform learning, or $a>0$ for variable rates. In both cases, the population evolves along a characteristic, though ultimately divergent, trajectory:

 \begin{itemize}
    \item \textit{Collapse of forgiveness}: From the initial uniform distribution on the strategy space, all agents rapidly move towards lower forgiveness $q$ values.
    \item \textit{Transient unforgiving state}: All agents reach minimum forgiveness $q = \eps$ and changes are transiently restricted to updates to agents' reciprocity parameters, $p_i$.
    \item \textit{Move to GTFT or ALLD}: after a sufficient buildup of high reciprocity $p$, the population moves towards TFT $(1-\eps, \eps$) and then to GTFT $(1-\eps,1-c/b-\eps)$; or, after sufficient buildup of low reciprocity $p$, the population converges on ALLD.
    \item \textit{Stability of GTFT}: The GTFT strategy $(1-\eps,1-c/b-\eps)$ is the only stationary distribution on the  boundary $p = 1-\eps$.
    \end{itemize}
Snapshots of the characteristic trajectory in a population with variable learning rates show all four stages (Fig.~\ref{fig:agent_steps}). Provided the variance of learning rates is sufficiently large ($a$ large enough), such a population will evolve to eventual adoption of  $\text{GTFT}=(1-\eps,1-c/b-\eps)$.

We study this characteristic trajectory step by step, with a special focus on the \textit{transient unforgiving state} and \textit{stability of GTFT}. The long-term fate of the population is settled at these two stages, and they are also the most amenable to analysis, because their dynamics are mostly restricted to a single dimension. Although we primarily discuss the discrete finite-$N$ model, the qualitative results are replicated in a continuous version of the model in the limit $N\rightarrow\infty$ (Fig.~S2 of the Supplemental Material). This indicates that the emergence of GTFT under variable learning rates does not depend on discretization.

\subsection{Collapse of Forgiveness}\label{sec:collapse-forgiveness}

\begin{figure*}[t]
\centering
\includegraphics[width=\textwidth]{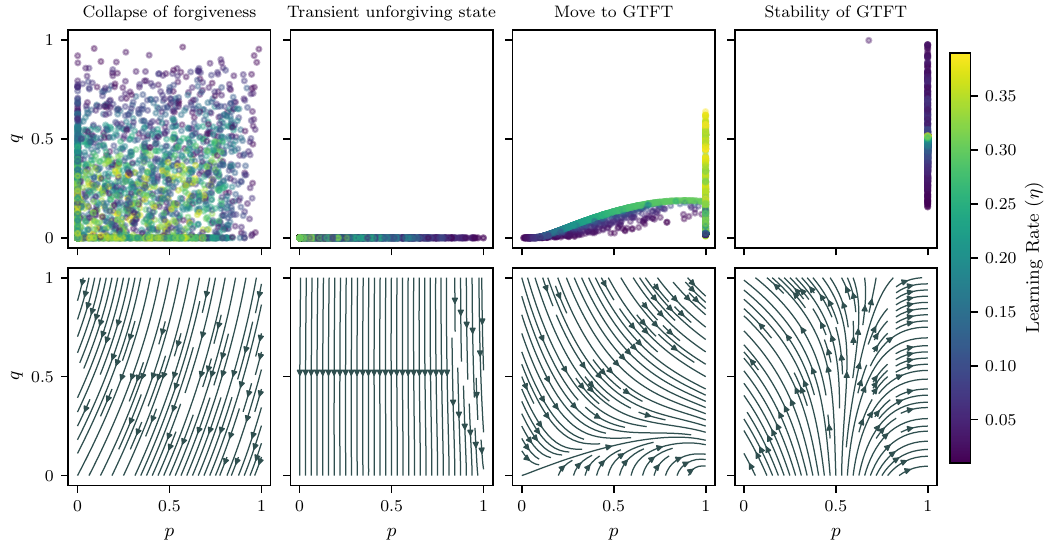}
\caption{\textbf{A population with variable learning rates passes through four stages before settling at generous tit-for-tat.} Agents' locations on the strategy space $[\eps, 1-\eps]^2$ at several learning steps (top, colored by learning rate $\eta$) and streamline plots of the vector field corresponding to the gradient against the population at these steps (bottom). \emph{Collapse of forgiveness}: early steps show stratification by learning rate; the slow learners with high reciprocity $p$ ultimately remain at high $p$ when forgiveness reaches its minimum, $q = \eps$. \emph{Transient unforgiving state}: the approximate beginning of the stage in which all agents $i$ have $q_i = \eps$ and some distribution of reciprocity $p_i$ not entirely at $\eps$. \emph{Move to GTFT}: agents move through the interior to reach a distribution of GTFT-like strategies. \emph{Stability of GTFT}: at the end of learning, the GTFT distribution is maintained.
}
\label{fig:agent_steps}
\end{figure*}

After initialization, faster learners move toward lower reciprocity $p$ and forgiveness $q$, reaching the unforgiving boundary $q=\eps$ while some slower learners retain high values of both coordinates (Fig.~\ref{fig:agent_steps}, \emph{Collapse of forgiveness}).
An explicit calculation against a uniform initial distribution shows that, at the parameters used here, both gradient components are negative throughout the strategy space (see Sec.~S2.1 of the Supplemental Material).
The initial decline in reciprocity $p$ therefore accompanies the decline in forgiveness $q$ even among strategies that later contribute to cooperation.

As the population changes, the gradients experienced by the remaining forgiving learners become nearly vertical in the simulations.
A population concentrated at ALLD provides an exact reference case: its gradient is $(-c\eps,-c(1-\eps))$ at every focal strategy, and a learner with rate $\eta$ loses less than $\eps(1+\eta c)$ in reciprocity before reaching $q=\eps$ (see Sec.~S2.2 of the Supplemental Material).
Although this bound assumes a fixed ALLD population, it explains how descent in a strongly defecting environment can preserve reciprocity.
Slower learners that begin with high $p$ can delay their descent until the gradient entails little further loss of $p$, leaving the population with support at high reciprocity upon entry to the transient unforgiving state (Fig.~\ref{fig:agent_steps}, \emph{Transient unforgiving state}).
With uniform learning rates, the simulated population instead enters this state with a lower maximum reciprocity (Fig.~S3 of the Supplemental Material, top row) and subsequently converges to ALLD (Fig.~\ref{fig:distributions}, left).

\subsection{Transient Unforgiving State}

The transient unforgiving state, in which all agents $i$ lie along the unforgiving line $q_i = \eps$, is reliably observed in both uniform and variable learning rate simulations.
Empirical distributions of $p$ values upon entering the transient unforgiving state are shown in Fig.~S3 of the Supplemental Material, top row. Although the variable learning rate population has more agents playing the ALLD strategy, it also has support on higher values of $p$, integral to ultimately stimulating cooperation.

The long-term fate of the population, either ALLD or GTFT, depends on the distribution of $p$ values (reciprocity) when the population first enters the transient unforgiving state. In fact, in simulations that enter the transient unforgiving state under variable learning rates, if we endow the agents with sufficiently low uniform learning rates from this point in time onward, the population still converges to GTFT. Conversely, starting from the transient unforgiving state in a simulation with uniform learning rates, and endowing agents with variable learning rates from that point onward, the population reliably ends up at ALLD. The distribution of $p$ values when reaching the transient unforgiving state is therefore the strongest determinant of the long-term outcome.

Because of these empirical results, we focus our analysis on the distribution of $p$ in a population when it first enters $q_i = \eps$ for all $i$.
We can determine which distributions over $p$ along the line $q_i= \eps$ converge to ALLD and which converge to TFT. When $q_i = \eps$ for all $i$, the gradient in the $p$ direction with respect to $(p_1, \eps)$ of its payoff against $(p_2, \eps)$ is given by
\[g_\eps(p_1, p_2) = \eps \frac{(1 + p_2 - \eps)\, \left( b\,(p_2 - \eps) - c\right)}{\left(1 - (p_1 -\eps)(p_2 - \eps)\right)^2}.
\]
We define the edge length $L = 1-2\eps,$ and let $p = p_1$ and $s = r_2 = p_2-\eps \in [0,L]$ denote the opponent's reciprocity measured from the defecting end of the edge. Then the $p$-component of the gradient of player 1 in a distribution defined by the probability measure $\mu$ is given by

\[G(p) = \eps \int_0^L \frac{(bs-c)(s+1)}{(1-(p-\eps)s)^2} \, d\mu_s(s) = \eps \, I_\mu(p).\]

The gradient in $p$ along the unforgiving line is smaller than the gradient in $q$ by a factor of order $\eps$, except near the boundary $p = 1-\eps$, where the two are similar.

The numerator of the integrand, $(bs-c)(s+1)$, is increasing in $s$, with a sign change at $\alpha = c/b$: opponents more reciprocal than $\alpha$ contribute a positive term to the gradient, and opponents less reciprocal than $\alpha$ a negative term. This sign change is the threshold described in Sec.~\ref{sec:gradient-ascent}: raising one's own cooperation pays only against an opponent whose degree of reciprocity exceeds $c/b$. Define the sign-separated tail integrals
\begin{align*}
P(\mu) &:= \int_\alpha^L (bs-c)(s+1)\, d\mu_s(s), \\
Q(\mu) &:= \int_0^\alpha (c-bs)(s+1)\, d\mu_s(s),
\end{align*}
which collect these two contributions with the denominator at its value for $p = \eps$, so that $\eps\, I_\mu(\eps) = \eps\,(P(\mu) - Q(\mu))$ is the gradient of an agent at ALLD, the least reciprocal strategy on the line. For the opposite fate we need every agent's gradient bounded above. The denominator factor $(1-(p-\eps)s)^{-2}$ is largest at $p = 1-\eps$, where it equals $(1-Ls)^{-2}$, so define
\begin{align*}
u(s) &:= \begin{cases} (bs-c)(s+1), & s < \alpha,\\[3pt] \dfrac{(bs-c)(s+1)}{(1-Ls)^2}, & s \ge \alpha, \end{cases} \\
U(\mu) &:= \int_0^L u(s)\, d\mu_s(s),
\end{align*}
the integrand of $I_\mu(p)$ with the denominator replaced, opponent by opponent, by the value that makes that opponent's contribution largest. The following theorem gives a sufficient condition for each fate, for arbitrary learning rates $\eta_i > 0$ and any step size.

\begin{theorem}
Let $\mu$ be the distribution of reciprocity $p$ when all agents have $q = \eps$. If $U(\mu) < 0$, the population converges to ALLD in finitely many learning steps. If $P(\mu) > Q(\mu)$ and the dynamics are restricted to $q = \eps$, the population converges to TFT in finitely many learning steps.
\label{thm:TUS}
\end{theorem}
\noindent A proof of the theorem can be found in Sec.~S3 of the Supplemental Material. The intuition for convergence to TFT is that an agent at ALLD, the least reciprocal strategy on the line, has gradient exactly $\eps\,(P(\mu) - Q(\mu))$, and that whenever this gradient is positive, every other agent's gradient is positive as well, and at least as large. This follows from a Chebyshev-type inequality for similarly ordered functions~\cite{wagener2006-ChebyshevsAlgebraicInequality}: the numerator $(bs-c)(s+1)$ and denominator factor $(1-(p-\eps)s)^{-2}$ increase with the opponent's degree of reciprocity $s$. The denominator factor also increases faster for a more reciprocal focal agent, so a more reciprocal agent's gradient weights the positive contributions of highly reciprocal opponents more heavily. If $P(\mu) > Q(\mu)$, one learning step therefore moves every agent toward higher reciprocity; since $(bs-c)(s+1)$ is increasing in $s$, this movement can only increase $P - Q$, so the updated population satisfies the same inequality and the next step again moves every agent toward TFT. The condition for ALLD is preserved in the same way: $U(\mu)$ bounds every agent's gradient from above, and because $u$ is increasing, movement toward lower reciprocity can only decrease $U$.

However, only the ALLD half of the theorem describes the full two-dimensional dynamics. Under $U(\mu)<0$, the unforgiving line is invariant: a negative gradient in $p$ implies a negative gradient in $q$, so the one-dimensional analysis is in fact a full description of the two-dimensional dynamics. Under $P(\mu)>Q(\mu)$, though, the line is provably not invariant: to reach TFT, the agents will all eventually need $s > c/b$, beyond which every agent's gradient in $q$ will be positive, so that the population leaves the $q=\eps$ line. Thus, this restriction is an idealization.

Theorem \ref{thm:TUS} resolves the fate of a wide range of distributions on $q = \eps$ (Fig.~S4 of the Supplemental Material). Given the parameters of a truncated Beta distribution, for example, a simple calculation determines whether the population goes to ALLD or TFT, or whether the distribution satisfies neither criterion and its fate is left undecided. Neither initial distribution of agents in the transient unforgiving state, with uniform or variable learning rates, provably reaches ALLD or TFT (even under restricted dynamics) via Theorem \ref{thm:TUS}. However, by step 1900, the uniform learning rate population has $U < 0$, guaranteeing its eventual convergence to ALLD. The variable learning rates population evolves towards TFT and then departs the $q = \eps$ axis before having $P > Q$.

\subsection{Evolution to GTFT}

From the $q = \eps$ axis, the populations that eventually move to GTFT require a sufficient buildup of TFT-like strategies with high reciprocity $p$. However, unlike in our study of the transient unforgiving state, which restricted analysis to the $q =\eps$ axis, in the simulations, agents take an indirect path from this axis through the interior of the square to the $p = 1-\eps$ axis (many distributions on which we loosely call ``GTFT-like''). That the population leaves the $q = \eps$ axis once all agents have $p-\eps>\alpha$ is established in Sec.~S3 of the Supplemental Material. We do not have a formal analysis of the movement in two dimensions with many agents, and so we provide a visual summary of the population's movement instead. We anticipate that, in most situations, a ``move to TFT'' when restricted to the $q = \eps$ axis, as in our analysis of the transient unforgiving state, would coincide with a move to GTFT if the agents were not restricted.

Before agents at ALLD move in the direction of TFT/GTFT, the agents at TFT start to move toward GTFT: some agents already occupy $p = 1-\eps$ with $q > \eps$ before any agents have left ALLD.
The direction in which agents at ALLD move changes over time: as more agents reach more cooperative strategies, gradient flows start to point toward higher $q$ values on the $p = 1-\eps$ line (Fig.~S5 of the Supplemental Material).
Faster learners soon approach GTFT $p = 1-\eps, q = 1-c/b - \eps$, while slower learners remain outside on higher and lower $q$ values. The initial distribution of $q$ once all agents attain $p = 1-\eps$ at step 5800 has mean $0.41$, median $0.40$, with standard deviation $0.13.$

\subsection{Stability of GTFT}\label{sec:gtft-stability}

GTFT is the only stationary distribution restricted to the boundary $p=1-\eps$ (Theorem~\ref{thm:GTFT}), and it persists under variable learning rates but not under uniform learning rates.  GTFT is not locally stable against perturbations to the interior, however. A population displaced from GTFT in a direction that lowers its $p-q$ cannot return to GTFT under uniform learning rates (Theorem~S2 of the Supplemental Material). For variable rates, however, perturbation from GTFT produces a slow transient followed by damped oscillations and a return to GTFT.

We begin by showing there is only one stationary distribution on the boundary $p=1-\eps$, namely GTFT, or $(1-\eps, 1-c/b-\eps)$. This is formalized in the following theorem.
\begin{theorem}
    When $\eps>0$ and $c/b < 1-2\eps$, the only distribution $\rho(\mathbf{x})$ on the $p = 1-\eps$ axis such that $\nabla_q \Pi(\mathbf{x}) =\int \nabla_q \pi(\mathbf{x}, \mathbf{y}) \, d\rho(\mathbf{y}) = 0$ for all $\mathbf{x} = (1-\eps, q)$ in the support of $\rho$ is a Dirac mass with frequency 1 at $q = 1-c/b-\eps$.
    \label{thm:GTFT}
\end{theorem} \noindent A proof of this theorem can be found in Sec.~S4 of the Supplemental Material. The intuition is that each opponent contributes to an agent's gradient in $q$ with a sign set by the opponent alone: positive if the opponent's degree of reciprocity exceeds $c/b$, and negative otherwise (Sec.~\ref{sec:gradient-ascent}). A less forgiving agent weights the positive contributions more heavily than a more forgiving agent does, so two agents at different values of $q$ cannot both have zero gradient. The only possibility is a population concentrated at the single value of $q$ where the degree of reciprocity equals $c/b$, that is, at GTFT.

The Dirac mass at GTFT is a true fixed point of the full two-dimensional dynamics: for a population concentrated at this point, both components of the gradient vanish, including the gradient in the $p$-direction. Both components vanish because GTFT has degree of reciprocity exactly $c/b$, so that an opponent at GTFT makes a co-player's payoff independent of the co-player's strategy (Sec.~\ref{sec:gradient-ascent}).

The GTFT outcome is nearly socially optimal, as well. Consider another fixed point of the system, ALLD, where the uniform learning rate agents end up. Two ALLD agents playing one another will each have payoff near zero, $(b-c)\eps$, while two GTFT agents will each have payoff $b(1-\eps) - c$. Although two ALLC agents have slightly higher payoffs $(b-c)(1-\eps)$, the difference is $O(\eps)$. GTFT maximizes the payoff of the population subject to being invasion-proof against less cooperative strategies~\cite{nowak1990-EvolutionStochasticStrategies}. ALLC, on the other hand, is vulnerable to invasion or gradient ascent exploitation.

However, over the course of the agent-based simulations, the population remains in a slowly moving transient state along this line. All gradients in the $q$-direction are positive for agents on the $p = 1-\eps$ line when they first reach it. In the limit of small $\eps$, there is a threshold $q_{\text{GTFT}} = 1-c/b$ such that if the mean $q$ value of a population on $p = 1-\eps$ exceeds $q_{\text{GTFT}}$, then an agent at ALLC = $(1-\eps, 1-\eps)$ will have negative gradient in the $p$-direction. The mechanism is that an agent facing forgiving opponents raises its payoff by lowering its own reciprocity. When the opponents an agent plays have forgiveness $q$ above $1-c/b$, that agent can defect and, with high probability, not be punished, and if it is punished the opponent's high forgiveness restores cooperation quickly. Equivalently, opponents on the reciprocal boundary with forgiveness above $1-c/b-\eps$ have degree of reciprocity below $c/b$, the threshold identified in Sec.~\ref{sec:gradient-ascent}.

Indeed, in simulations, once the population approaches a mean $q$ value of $1-c/b$ (here $1/2$), multiple agents with high $q$ values then move toward lower $p$ (Fig.~\ref{fig:distributions}). In response, most of the population lowers their $q$ values, reducing the payoff from exploitation. Most agents who departed the $p = 1-\eps$ axis return, except for one. There is a back and forth as this one agent becomes more or less exploitative as the rest of the population moves in response. This back and forth can be seen in Fig.~\ref{fig:distributions} ``Variable Learning Rates.'' As mean $q$ approaches $1-c/b$, the minimum value of reciprocity $p$ becomes less than $1-\eps$, immediately followed by a slight decrease in the mean forgiveness $q$. The population eventually reaches a distribution where the majority of agents are near $q = 1-c/b-\eps,$ the stationary point. Given enough simulation time the population converges to GTFT. A long-term simulation of a reduced population sampled from the final state of our system suggests that the system is indeed slowly converging on GTFT (Fig.~S6 of the Supplemental Material).

Although GTFT is a stationary point, it is not locally stable. Perturbing even one agent from GTFT can destabilize the  population. When slightly increasing the $q$ value of a population at GTFT, a population with uniform learning rates provably never returns to GTFT (see Sec.~S6 of the Supplemental Material). In simulations this population reaches the $q=\eps$ axis (Fig.~S7 of the Supplemental Material), from which Theorem \ref{thm:TUS} guarantees convergence to ALLD, since its degree of reciprocity remains below $c/b$. However, a population with variable learning rates, despite being locally destabilized, eventually recovers to GTFT through a transient. Similarly to the main experiment, the preservation of high $p$ values by relatively slow learners stimulates the population's return to the $p = 1-\eps$ axis. This result is confirmed by perturbing random fractions of the population, with displacements having random magnitudes and directions (independent of other agents), over 500 simulations per learning rate regime: uniform learning rate populations never recover $p = 1-\eps$ for any agent, but variable learning rate populations always recover $p = 1-\eps$ for all agents (Fig.~S8 of the Supplemental Material).

\subsection{Robustness to population size and benefit-to-cost ratio}\label{sec:robustness}

The emergence of cooperation under variable learning rates is not particular to one population size or one payoff regime (Fig.~S9 of the Supplemental Material). We varied the benefit-to-cost ratio $b/c$ and the degree of heterogeneity in learning rates, and recorded the fraction of replicate populations that reached the cooperative GTFT outcome, for populations of size $N = 10$, $100$, and $1000$. Three patterns emerge. First, cooperation becomes more frequent as the benefit-to-cost ratio increases, consistent with the threshold role of $c/b$ in Theorem \ref{thm:TUS} and in the location of GTFT itself. Second, greater heterogeneity in learning rates shifts the boundary of cooperation toward smaller $b/c$: populations with more variable learners can sustain cooperation in games where less variable populations cannot. The learning rates are distributed according to a symmetric Beta distribution with parameter $\alpha$: $\eta_i \sim \mathrm{Beta}(\alpha, \alpha)$, shifted and scaled onto the range $[0.01, 0.39]$. As $-\log\alpha$ increases, the distribution of learning rates increases in variance, going from a point mass at $\eta = 0.2$ to a uniform distribution on $[0.01, 0.39]$ to a two-point distribution with equal probability mass at $0.01$ and $0.39$. Although some $b/c$ ratios cannot sustain cooperation with any parameter $\alpha$, some intermediate $b/c$ ratios can only sustain cooperation with sufficiently high learning rate variance. Third, the transition sharpens with population size. In small populations ($N = 10$) outcomes are mixed, with small-to-intermediate fractions of replicates reaching cooperation across much of the parameter range; in large populations ($N = 1000$) the outcome is nearly deterministic, with a sharp boundary in the $(b/c, \text{heterogeneity})$ plane separating populations that reach ALLD from populations that reach GTFT. This is consistent with our analysis: the fate of the population is determined by the distribution of reciprocity on entering the transient unforgiving state, and in larger populations that distribution concentrates around its expectation, so replicate populations behave alike.

\section{Discussion}

We have studied a population of selfish, myopic learners who update reactive IPD strategies by simultaneous gradient ascent on their payoffs against the whole population. The fate of such a population can hinge on a single naturalistic parameter: variation among individuals in their rates of learning. With uniform learning rates the population converges to unconditional defection, whereas with variable rates it converges to a stable distribution of generous tit-for-tat strategies and sustains nearly socially optimal payoffs. Along the way, every population passes through a transient unforgiving state in which forgiveness has collapsed and the distribution of reciprocity determines the long-term outcome. We proved criteria on this distribution that guarantee convergence to ALLD or to TFT (Theorem \ref{thm:TUS}), and we proved that GTFT is the only stationary distribution on the reciprocal boundary of the strategy space (Theorem \ref{thm:GTFT}). These results are robust to population size and benefit-to-cost ratio, with the transition between defection and cooperation sharpening as the population size grows (Fig.~S9 of the Supplemental Material).

Our results identify a condition under which gradient ascent solves a social dilemma, and so they speak directly to the literature on optimization and multi-agent learning. Gradient methods are effective across a wide range of optimization and learning problems~\cite{singh2000-NashConvergenceGradient,zinkevich2003-OnlineConvexProgramming}, and yet strict gradient ascent typically fails on social dilemmas such as the IPD, where pairs of learners making selfish updates converge to mutual defection~\cite{sandholm1996-MultiagentReinforcementLearning,macy2002-LearningDynamicsSocial,tuyls2006-EvolutionaryDynamicalAnalysis,bloembergen2015-EvolutionaryDynamicsMultiAgent}. Existing remedies modify the learner: agents may weight group payoffs in their gradients~\cite{zhang2019-SAIGAMultiagentReinforcement}, anticipate the learning of their opponents~\cite{foerster2018-LearningOpponentLearningAwareness}, or adapt their learning rate to their performance~\cite{bowling2002-MultiagentLearningUsing}. We found a way to make gradient ascent succeed that requires no engineering at all: place the learners in a population and let their learning rates vary. The agents remain purely selfish and myopic; only the composition of the population changes.

Heterogeneity in learning rates is not the only way to avoid the large basin of attraction to ALLD. Randomness in the learning process can have a similar effect. In our simulations, a population with uniform learning rates also converges to a GTFT-like distribution under gradient ascent perturbed by small isotropic noise (Fig.~S10 of the Supplemental Material), although an isolated pair of noisy learners still converges to mutual defection (Fig.~S11 of the Supplemental Material). Noise of this kind could enter through the strategy update directly, or through agents estimating their gradient from a finite sample of opponents rather than the whole population.
The value of variable learning rates is therefore not that they are the unique cause of cooperation, but that they are a more mechanistic and biologically realistic one. Variation in learning rates generates no new strategies and posits nothing beyond a fixed, intrinsic trait that differs across individuals, the kind of variation organisms are known to exhibit. Variation in learning rates therefore preserves the diversity already present in the population: slow learners with high reciprocity survive the collapse of forgiveness with their reciprocity intact, and they leave the population holding the high-reciprocity strategies that the later move to GTFT requires.

This contrast is sharpest against evolutionary game theory on a finite set of IPD strategies. When a population of ALLD, TFT, and ALLC strategies evolves by payoff-biased imitation with rare mutation, it never settles into cooperation. The population instead cycles indefinitely between cooperative and defecting phases: tit-for-tat permits the rise of generous players, then unconditional cooperators, who are then exploited by defectors, and the population collapses back toward ALLD before the cycle repeats~\cite{nowak1992-TitTatHeterogeneous,imhof2005-EvolutionaryCyclesCooperation,imhof2010-StochasticEvolutionaryDynamics,nowak2006-EvolutionaryDynamicsExploring}. Stable, persistent cooperation is not an outcome of this dynamic. Learning in a population with variation in learning rates succeeds where this imitation dynamic fails. The two differ in mechanism, in that our agents compute a local payoff gradient over a continuous strategy space rather than copying successful neighbors from a finite menu. The feedback we observe near GTFT, where an occasional agent moves toward exploitation and the population's forgiveness decreases in response, is a damped perturbation around this stable outcome, not a recurrence of the cooperative and defecting phases.

Closer to our setting, McAvoy et al.~\cite{mcavoy2022-SelfishOptimizationCollective} showed that selfish gradient learners can reach social optima when updates occur through rare, random pairwise encounters in a large population. Their mechanism is different from ours: optimality there arises from the ephemerality of interactions, whereas our agents all update simultaneously against the entire population, and optimality arises from heterogeneity in learning rates. The two models nevertheless agree on a broader point, namely that the population context can reverse the dysfunction of selfish learning in pairs.

Several limitations suggest directions for future work. Our agents compute gradients of their payoff against the entire population; a more realistic learner would sample a subset of opponents, introducing stochasticity whose effect on whether the population reaches ALLD or GTFT is unknown.
We restricted attention to reactive strategies; the full space of memory-one strategies doubles the dimension and includes strategies, such as extortionate ones, that may alter the transient unforgiving state. The continuous-time limit of our population dynamics is a nonlocal transport equation (Sec.~S1.2 of the Supplemental Material), and a rigorous treatment of its well-posedness and long-term behavior, including mass accumulation on the boundary, remains open. Finally, our agents' learning rates are fixed for life; learning rates that decay over time, as in stochastic approximation~\cite{singh2000-NashConvergenceGradient,bowling2002-MultiagentLearningUsing}, would likely stabilize the slow transient near GTFT into true convergence.

Our results also hint at broader implications for multi-agent learning.
The distribution of learning rates (or other microscopic parameters) in a population might determine which collective outcomes become accessible, and which are resilient.
Coordination problems, in particular, might benefit from this heterogeneity.
Mismatches are costly, so fast learners might abandon optimal strategies, but slower learners could maintain them long enough for others to relearn them.
Preserving strategic diversity through heterogeneity might lead to more optimal outcomes for multiple agents across a broad class of games and tasks.

Individuals differ in how quickly they learn, both across and within species, and that variation is well documented~\cite{thornton2012-IndividualVariationCognitive,sih2012-LinkingBehaviouralSyndromes,frank2007-GeneticTripleDissociation}. Variation in learning rates might also be thought of as a measure of one's reactivity to a situation: how quickly and how much should they act upon information at their disposal? Should they change their entire strategy, or wait for more to play out? It is similarly naturalistic that this trait might vary in a population. Our results show that this mundane form of variation is sufficient to produce stable and nearly optimal cooperation in a population of entirely selfish learners, with no genetic predisposition to altruism, no foresight, and no randomness in their updates. Variable learning rates cannot create high reciprocity where none exists; what they do is preserve the strategic diversity a population starts with, long enough for reciprocity to produce cooperation. Diversity in how individuals learn, not only in what they do, may thus be an underappreciated resource for the emergence of cooperation.

\section*{Data availability}
The code used to run the simulations, generate the figures, and perform the numerical checks in this article is openly available at \href{https://github.com/james-wolfe/reactive-learning}{https://github.com/james-wolfe/reactive-learning}.
The simulation data are archived on Zenodo~\cite{wolfe2026-Data}.

\section*{Acknowledgments}
We thank the members of the Plotkin lab for helpful discussions.

{\small
}

\onecolumn
\newgeometry{margin=1in}

\makeatletter
\setlength{\@fptop}{0pt}
\setlength{\@fpsep}{12pt plus 1fil}
\setlength{\@fpbot}{0pt plus 1fil}
\renewcommand{\p@subsection}{}
\makeatother
\renewcommand{\thefigure}{S\arabic{figure}}
\renewcommand{\thetable}{S\arabic{table}}
\renewcommand{\theequation}{S\arabic{equation}}
\renewcommand{\thesection}{S\arabic{section}}
\renewcommand{\thesubsection}{\thesection.\arabic{subsection}}
\renewcommand{\thetheorem}{S\arabic{theorem}}
\renewcommand{\thepage}{S\arabic{page}}
\renewcommand{\theHfigure}{S\arabic{figure}}
\renewcommand{\theHtable}{S\arabic{table}}
\renewcommand{\theHequation}{S\arabic{equation}}
\renewcommand{\theHsection}{S\arabic{section}}
\renewcommand{\theHtheorem}{S\arabic{theorem}}
\setcounter{figure}{0}
\setcounter{table}{0}
\setcounter{equation}{0}
\setcounter{section}{0}
\setcounter{theorem}{0}
\setcounter{page}{1}

\begin{center}
\section*{Supplemental Material for\\ \scititle}
James P.~Wolfe and Joshua B.~Plotkin$^\ast$\\
\small$^\ast$Corresponding author. Email: jplotkin@sas.upenn.edu\\
\end{center}

\section{Infinite population limit}\label{sec:infinite-population}

\subsection{Formulation}

Much like in the agent-based formulation, for any individual strategy $x \in [\eps, 1-\eps]^2,$ its expected payoff at time $t$ is given by the payoff of $x$ against all other strategies in $[\eps,1-\eps]^2$, weighted by their relative frequency:

\[
\Pi^t(x) = \int \pi(x, {y}) \, dF^t({y}).
\]
The gradient of strategy $x$ playing against the distribution $F^t(y)$ is then
\[
\nabla_{x} \, \Pi^t(x) = \int \nabla_{x} \pi(x, {y}) \, dF^t({y}).
\]
From this gradient we construct an update rule for the population. First, denote the strategy space $X = [\eps, 1-\eps]^2$ and let $E \subset \mathbb{R}^+ $ be the set of possible learning rates $\eta$ in the population. To track both strategies and learning rates at time $t$, we bundle them into the joint distribution $\mu^t$ over $\Omega = X\times E$. The function $\Phi^t$ returns the location after one gradient ascent step of the point $(x, \eta)$, with

\[\Phi^t(x, \eta) = (\mathrm{Proj}_X(x + \eta \, \nabla_{x} \, \Pi^t(x)), \eta).\]

The joint distribution update for the population is then given by the pushforward measure (a natural discrete-time analogue of a continuity equation)

\[\mu^{t+1}= (\Phi^t)_\# \mu^t,\]
i.e. for any Borel set $A \subset \Omega,$
\[\mu^{t+1}(A) = \iint \mathbf{1}_{\Phi^t(x,\eta)\in A} \, \mu^t(dx, d\eta).\]
This definition provides an exact but implicit description of the update rule. It says, in essence: take the measure of all points such that after gradient update, the point falls within $A$. The total measure of all these points is the new measure of $A$. Note that this definition allows for continuous probability densities, as well as the accumulation of point masses on the boundaries of $X$. In general, the update rule cannot be simplified to a closed-form density solution, excepting special cases.

\subsection{Continuous-time limit}

The discrete-time pushforward above has a formal continuous-time limit. Introduce a step size $h > 0$, replace the update by $\Phi^t_h(x,\eta) = (\mathrm{Proj}_X(x + h\,\eta\, \nabla_x \Pi^t(x)), \eta)$, and rescale time as $\tau = h t$. As $h \to 0$, the sequence of pushforward updates converges formally to the nonlocal continuity (transport) equation
\[\partial_\tau \mu_\tau + \nabla_x \cdot \big(\mu_\tau \, \eta\, v[\mu_\tau]\big) = 0, \qquad v[\mu](x) = \nabla_x \int_\Omega \pi(x, y)\, \mu(dy, d\eta'),\]
on $\Omega = X \times E$, with no-flux boundary conditions on $\partial X$ inherited from the projection, which permit the accumulation of mass on the boundary. We conjecture that the behavior established in the discrete-time model persists in this limit: depending on the distribution of reciprocity at the entry to the transient unforgiving state, solutions converge either to a Dirac mass at ALLD or to the Dirac mass at GTFT, which are the analogues of Theorems~1 and~2 of the main text. A rigorous treatment of existence and uniqueness for this equation, with boundary mass accumulation, is left for future work.

\subsection{Implementation}

Since the distribution of learning rates is unchanging in time, we can handle the mass associated with each value of $\eta$ separately. Consider $L$ separate $M\times N$-sized grids of equally spaced points on $[\eps, 1-\eps]^2$, each grid $l$ corresponding to a different learning rate $\eta_l$. Each point on grid $l$ has an associated frequency $f_l{(p_m,q_n)}$ such that $\sum_{m, n} f_l(p_m, q_n) = 1/L$. That way, each learning rate $\eta_l$ will always be associated with the same amount of mass.

The gradient of point $x = (p_i, q_j)$ need only have its notation adjusted: it is still given by the weighted average, across all grids, of the gradient of its payoff against the rest of the population,
\[\nabla_{x} \, \Pi^t(x) = \sum_{l=1}^L\sum_{n=1}^N \sum_{m=1}^M f_l(p_m, q_n)  \nabla_{x} \Pi(x, (p_m, q_n)).\]
The function $\Phi_k^t$ will now send a point $x$ one step along its gradient with learning rate $\eta_k$: \[\Phi_k^t(x)= \mathrm{Proj}_{[\eps, 1-\eps]^2}(x + \eta_k\, \nabla_{x} \, \Pi^t(x)).\]

After one gradient step, the mass associated with point $x$ could be anywhere in $[\eps, 1-\eps]^2$. To redistribute the mass back to the preset grid points, we implement a bilinear interpolation scheme, outlined below.

Let $(\tilde{p_i}, \tilde{q_j}) = \Phi_k^t(p_i, q_j)$ and \[a = \min\{\max\{r: p_r \leq \tilde{p_i}\}, M-1\}, \, b = \min\{\max\{s: q_s\leq \tilde{q_j}\}, N-1\}.\]
Then $p_a$ and $q_b$ are the maximum valued grid points such that $\tilde{p_i} \geq p_a$ and $\tilde{q_j} \geq q_b,$ except, by construction, $a$ and $b$ cannot be the highest possible index (which might happen if $(\tilde{p_i}, \tilde{q_j})$ was on the edge of the space). Now, consider the four points $(p_a, q_b), (p_{a+1}, q_b), (p_a, q_{b+1}), (p_{a+1}, q_{b+1})$. These form a square that either contains $(\tilde{p_i}, \tilde{q_j})$ in its interior or on its boundary. Let
\[d_p = \frac{\tilde{p_i}-p_a}{p_{a+1}-p_a}, \, d_q = \frac{\tilde{q_j}-q_b}{q_{b+1}-q_b}\]
denote the relative distances in each direction to the $p$ and $q$ coordinates of the square's corners.
If $w_{r,s}(\tilde{p_i},\tilde{q_j})$ is the fraction of the mass at $(p_m, q_n)$ that is allocated to the point $(p_r, q_s)$, then our mass distribution scheme is the following: \begin{align*}
w_{a,b}(\tilde{p_i},\tilde{q_j}) &= (1-d_p)(1-d_q), \\
w_{a+1,b}(\tilde{p_i},\tilde{q_j}) &= d_p(1-d_q), \\
w_{a,b+1}(\tilde{p_i},\tilde{q_j}) &= (1-d_p)d_q, \\
w_{a+1,b+1}(\tilde{p_i},\tilde{q_j}) &= d_p d_q,
\end{align*} and $w_{r,s} = 0$ for all other $r,s$. Note that this handles edge cases (like boundaries) well. (A geometric view of this interpolation scheme is that the fraction of mass that goes to one corner is equal to the area of the rectangle between the point in question and the opposite corner, divided by the area of each grid square.)

Finally, we write the updated frequency for the point $(p_m, q_n)$ on the $k$th grid, which is given by a similar pushforward:

\[f_k^{t+1}(p_m, q_n) = \sum_{i, j} f_k^t(p_i, q_j)\,  w_{m, n} (\Phi_k^t(p_i, q_j)).\]

\section{Gradients during the collapse of forgiveness}
\label{sec:collapse-gradients}

Variation in learning rates can preserve reciprocity by delaying slow learners' descent toward the unforgiving line.
Slow learners move little while the initial gradient points toward lower reciprocity and forgiveness, and can retain high reciprocity until their gradients become nearly vertical.
Two reference populations clarify this timing mechanism: the uniform initial distribution and a population concentrated at ALLD.

\subsection{The uniform initial distribution}
\label{sec:unif-gradient-collapse}

Write $X=[\eps,1-\eps]^2$ and $L=1-2\eps$.
Under the uniform distribution on $X$, the opponent's degree of reciprocity $r=p_2-q_2$ has density $(L-|r|)/L^2$, and its mean forgiveness conditional on $r$ is $(1-r)/2$.
Averaging the pairwise gradient from Sec.~II\,B of the main text therefore gives
\[
\begin{aligned}
G_p(p,q)
&=\frac{1}{L^2}\int_{-L}^{L}
\frac{(br-c)[\frac12+r(q-\frac12)]}
     {(1-(p-q)r)^2}(L-|r|)\,dr,\\
G_q(p,q)
&=\frac{1}{L^2}\int_{-L}^{L}
\frac{(br-c)[\frac12+r(\frac12-p)]}
     {(1-(p-q)r)^2}(L-|r|)\,dr.
\end{aligned}
\]
These exact integrals admit closed forms.
Along each segment with $p-q$ fixed, both components are linear in $(p,q)$, so their maxima occur at the segment's endpoints.
Their maxima over $X$ can therefore be found by numerical optimization along its edges.

At $b=1$, $c=1/2$, and $\eps=10^{-3}$, the largest component is approximately $-0.157536$, attained at $\mathrm{TFT}=(1-\eps,\eps)$.
Both components are therefore negative throughout the square.
Numerical optimization gives directions between $26.67^\circ$ and $63.33^\circ$ to the left of the downward vertical: initial learning entails substantial decreases in both reciprocity and forgiveness.
These statements describe the unprojected gradient against the exactly uniform distribution; projection limits decreases at the boundary, and a finite sample has its own empirical gradient.

\subsection{A population concentrated at ALLD}
\label{sec:alld-gradient-collapse}

Against a population $F_{\mathrm D}$ concentrated at $\mathrm{ALLD}=(\eps,\eps)$, substituting $r_2=0$ and $q_2=\eps$ into the pairwise gradient gives
\begin{equation}
G_{F_{\mathrm D}}(\mathbf{x})
=\bigl(-c\eps,-c(1-\eps)\bigr)
\qquad\text{for every }\mathbf{x}\in X.
\label{eq:alld-field}
\end{equation}
This constant field points only $\arctan(\eps/(1-\eps))\approx0.057^\circ$ left of the downward vertical at $\eps=10^{-3}$.
Before projection, forgiveness therefore decreases $999$ times as rapidly as reciprocity, independently of the learner's strategy or learning rate.
A learner with learning rate $\eta$ and initial forgiveness $q_0$ reaches $q=\eps$ after $k=\lceil (q_0-\eps)/(\eta c(1-\eps))\rceil$ learning steps, during which its reciprocity decreases by at most $k\,\eta c\eps < \eps(q_0-\eps)/(1-\eps)+\eta c\eps < \eps(1+\eta c)$, the bound quoted in Sec.~III\,A of the main text.

These calculations support the timing mechanism observed in Sec.~III\,A of the main text.
Slow learners retain high $p$ and $q$ during the initial diagonal descent, while faster learners accumulate near the unforgiving boundary.
Their later descent occurs under nearly vertical gradients, allowing them to reach low forgiveness with high reciprocity still intact (Fig.~3 of the main text, \emph{Collapse of forgiveness}).

\section{Proof of Theorem~1 of the main text: convergence in the transient unforgiving state}

On the edge $q_i = \eps$ the population is described by the distribution of reciprocity values $p \in [\eps, 1-\eps]$, and each agent moves by the mean gradient $G(p) = \eps\, I_\mu(p)$. Write $s = p - \eps \in [0, L]$ with $L = 1-2\eps$, let $\mu_s$ denote the induced law of $s$, and write $\phi(s) = (bs-c)(s+1)$ for the numerator of the integrand and $h(p,s) = (1-(p-\eps)s)^{-2}$ for the denominator factor, so that
\[
I_\mu(p) = \int_0^L \phi(s)\, h(p,s)\, d\mu_s(s), \qquad I_\mu(\eps) = P(\mu) - Q(\mu) = \int_0^L \phi\, d\mu_s .
\]
Agent $i$, with learning rate $\eta_i > 0$, updates by $p_i \mapsto \operatorname{Proj}_{[\eps,1-\eps]}\big(p_i + \eta_i\, G(p_i)\big)$. The argument has an invariance lemma (Step 0) plus three steps for each of the two fates, ALLD and TFT: a bound on $I_\mu(p)$ that holds for every $p$ at once, from below for TFT and from above for ALLD (Step 1); the observation that the bound is preserved when every agent moves in the same direction (Step 2); and the conversion of a uniform sign of the gradient into convergence in finitely many steps (Step 3). No assumption on the step sizes is needed, and the learning rates may differ across agents, because Step 2 uses only the direction in which each agent moves, not the preservation of the agents' order.

\paragraph{Step 0: the $q$-direction on the unforgiving line.}
On the line $q=\eps$, the gradient of agent $i$'s payoff in the $q$ direction is
\[
  G_q(p_i) \;=\; \int_0^L (bs-c)\,\big(1-\eps-p_i s\big)\, h(p_i,s)\, d\mu_s(s),
\]
whose integrand differs from that of $\eps I_\mu(p_i)$ only in the positive factor $1-\eps-p_i s \ge \eps$.
Writing $\rho_i(s) := (1-\eps-p_i s)\big/\big(\eps(1+s)\big)$, which is positive and strictly decreasing on $[0,L]$ since $\tfrac{d}{ds}\rho_i \propto -\eps(1+p_i-\eps) < 0$, we have $G_q(p_i) = \int \eps\,\phi(s)\,h(p_i,s)\,\rho_i(s)\,d\mu_s(s)$.
Because $\phi$ and $\rho_i - \rho_i(\alpha)$ are oppositely ordered, the integrand is bounded pointwise by $\eps\,\phi(s)h(p_i,s)\rho_i(\alpha)$, and integrating,
\begin{equation}
  G_q(p_i) \;\le\; \rho_i(\alpha)\, G(p_i), \qquad \rho_i(\alpha) > 0 .
  \label{eq:qbound}
\end{equation}
Hence $G(p_i) < 0$ implies $G_q(p_i) < 0$: an agent whose reciprocity is falling is also lowering its forgiveness, and the projection holds it at $q = \eps$.
Conversely, if $s > \alpha$ throughout the support of $\mu_s$, every factor in the integrand of $G_q$ is positive and $G_q(p_i) > 0$ for every $i$.

\paragraph{Step 1, TFT: a lower bound on every agent's gradient.} For every $p$,
\begin{equation}
I_\mu(p) \;\ge\; h(p,\alpha)\,\big(P(\mu) - Q(\mu)\big), \qquad h(p,\alpha) = \frac{1}{(1-(p-\eps)\alpha)^2} \;\ge\; 1 .
\label{eq:cheb}
\end{equation}
To see this, write $I_\mu(p) - h(p,\alpha)\, I_\mu(\eps) = \int_0^L \phi(s)\,\big[h(p,s) - h(p,\alpha)\big]\, d\mu_s(s)$. For fixed $p$ the factor $h(p,s)$ is nondecreasing in $s$, so $h(p,s) - h(p,\alpha)$ has the sign of $s - \alpha$; and $\phi(s)$ has the sign of $s - \alpha$ as well. The integrand is therefore nonnegative, which gives \eqref{eq:cheb}. (This is a Chebyshev-type inequality for similarly ordered functions~\cite{wagener2006-ChebyshevsAlgebraicInequality}.) In particular, if $P(\mu) > Q(\mu)$ then $I_\mu(p) \ge P(\mu) - Q(\mu) > 0$ for every $p$: whenever the gradient at ALLD, $p = \eps$, is positive, so is every agent's.

\paragraph{Step 1, ALLD: an upper bound on every agent's gradient.} For the ALLD fate we need the opposite bound, an upper bound on $I_\mu(p)$ that holds for every $p$. For fixed $s$, $h(p,s)$ increases from $1$ at $p = \eps$ to $(1-Ls)^{-2}$ at $p = 1-\eps$. Where $\phi(s) < 0$, that is for $s < \alpha$, the contribution $\phi(s)\,h(p,s)$ is largest at $p = \eps$; where $\phi(s) \ge 0$ it is largest at $p = 1-\eps$. Hence $\phi(s)\,h(p,s) \le u(s)$ for every $p$ and $s$, with $u$ as defined in the main text, and integrating against $\mu_s$,
\begin{equation}
I_\mu(p) \;\le\; U(\mu) \qquad \text{for every } p .
\label{eq:upper}
\end{equation}
Moreover $u$ is continuous and nondecreasing on $[0,L]$: on $[0,\alpha)$, $u'(s) = 2bs + (b-c) > 0$; on $[\alpha,L]$, $u$ is the product of the nonnegative increasing function $\phi$ and the positive increasing function $(1-Ls)^{-2}$; and both pieces vanish at $s = \alpha$.

\paragraph{Step 2: each bound is preserved by one learning step.} For TFT, suppose $P(\mu_0) > Q(\mu_0)$ and set $\delta = P(\mu_0) - Q(\mu_0)$. By \eqref{eq:cheb} every agent has $G(p_i) \ge \eps\delta > 0$, so every agent moves weakly toward higher reciprocity, $s_i' \ge s_i$, whatever its learning rate. Here the restriction to $q = \eps$ is used: without it, $s_i = p_i - q_i$ need not increase even though $p_i$ does. Since $\phi$ is increasing on $[0,L]$, with $\phi'(s) = 2bs + (b-c) > 0$,
\[
P(\mu_1) - Q(\mu_1) = \frac{1}{N}\sum_i \phi(s_i') \;\ge\; \frac{1}{N}\sum_i \phi(s_i) = P(\mu_0) - Q(\mu_0) = \delta,
\]
and in the Eulerian model the same inequality holds with the sums replaced by integrals against $\mu_0$. By induction $P(\mu_t) - Q(\mu_t) \ge \delta$ for every $t$, hence $I_{\mu_t}(p) \ge \delta$ for every agent at every step.

Suppose instead $U(\mu_0) = -\delta < 0$. By \eqref{eq:upper} every agent has $G(p_i) \le -\eps\delta < 0$, and then by \eqref{eq:qbound}, every agent also has $G_q<0$, so the population remains on the line and $\mu_1$ is again a distribution on $[0, L]$. Thus, the population moves weakly toward lower reciprocity, $s_i' \le s_i$. Since $u$ is nondecreasing, $U(\mu_1) = \frac{1}{N}\sum_i u(s_i') \le \frac{1}{N}\sum_i u(s_i) = U(\mu_0)$, and by induction $I_{\mu_t}(p) \le -\delta$ for every agent at every step.

\paragraph{Step 3: convergence in finitely many steps.} In the TFT case each agent gains at least $\eta_i\,\eps\,\delta$ in reciprocity per step until it reaches $p = 1-\eps$, where the projection holds it, so every agent is at $1-\eps$ after at most $\lceil L / (\min_i \eta_i\, \eps\, \delta) \rceil$ steps, the population at TFT $(1-\eps, \eps)$. In the ALLD case each agent loses at least $\eta_i\,\eps\,\delta$ per step until it reaches $p = \eps$, and the population is at ALLD $(\eps,\eps)$ after the same number of steps.

Finally, the TFT condition is exact for an argument that certifies a positive gradient at every point of the line: if $P(\mu) \le Q(\mu)$, then $I_\mu(\eps) = P(\mu) - Q(\mu) \le 0$, so an agent at ALLD, $p = \eps$, does not have a positive gradient. The ALLD condition is a sufficient condition only: $U(\mu)$ is an upper bound that in general no single agent attains, because the denominator factor is evaluated at $p = \eps$ for opponents less reciprocal than $\alpha$ and at $p = 1-\eps$ for opponents more reciprocal than $\alpha$. \hfill$\square$

Once every agent's $s$ exceeds $\alpha$, every agent's gradient in $q$ is positive, meaning the population will leave the $q = \eps$ line. This can be seen from the equation for $G_q(p_i)$ in Step 0. When $P(\mu) > Q(\mu)$, every agent moves monotonically to $L>\alpha$, implying that at some point before (with sufficiently small learning rate) reaching TFT, the entire population will have positive gradient in $q$. The restriction that $q =\eps$ is thus necessary; the population would assuredly leave the line without it.

\section{Proof that GTFT is the only stationary distribution on the reciprocal line}

\subsection*{The gradient along the reciprocal line}

Fix $p = 1-\eps$ and   consider a focal agent $(p,q)$ playing an opponent $(p,y)$, both on the
reciprocal boundary. Their payoff is
\[\pi(q, y) = \frac{b(q(p-y)+y) - c(q + (p-q)y)}{1-(p-q)(p-y)},\]
and differentiating with respect to the focal agent's forgiveness $q$ gives
\begin{equation}
\partial_q \pi(q,y) = \frac{\psi(y)}{\big(1-(p-q)(p-y)\big)^{2}},
\qquad
\psi(y) = (1-p)(1+p-y)\big(b(p-y)-c\big).
\label{eq:qgrad}
\end{equation}
The numerator $\psi(y)$ depends only on the opponent's forgiveness $y$, and the focal agent's own forgiveness $q$ enters only through the denominator, which is strictly positive.
The sign of the pairwise gradient is therefore set entirely by the opponent. Against the population $\rho$ the mean gradient of an agent at $q$ is
\[H(q) \;=\; \int \psi(y)\, \varphi(q,y) \, d\rho(y),
\qquad \varphi(q,y) = \frac{1}{\big(1-(p-q)(p-y)\big)^{2}}.\]

\subsection*{Only the Dirac mass at GTFT is stationary}

\begin{theorem}[Theorem~2 of the main text, restated]
Let $\eps>0$, $b>c>0$, and $c/b < 1-2\eps$.
Let $\rho$ be any Borel probability measure on the line $p = 1-\eps$, atomic or continuous, such that $H(q) = 0$ for every $q$ in the support of $\rho$.
Then $\rho$ is the Dirac mass at $q^{*} = 1-c/b-\eps$, i.e. at GTFT.
\end{theorem}

\begin{proof}
Write $r(y) = p-y \in [0, 1-2\eps]$ for the opponent's degree of reciprocity, so that $1-(p-q)(p-y) = 1 - r(q)r(y)$ with $0 \le r(q)r(y) \le (1-2\eps)^2 < 1$.
On this range the kernel admits the uniformly convergent expansion
\[
    \varphi(q,y) \;=\; \frac{1}{\big(1-r(q)r(y)\big)^{2}} \;=\; \sum_{n=0}^{\infty}(n+1)\,r(q)^{n}\,r(y)^{n}.
\]
Because $H$ vanishes at every point of the support of $\rho$, the integral of $H(q)\psi(q)$ against $\rho$ is zero.
Substituting the expansion and exchanging the sum with the two integrals, which uniform convergence permits,
\[
    0 \;=\; \iint \psi(q)\,\varphi(q,y)\,\psi(y)\; d\rho(q)\, d\rho(y)
    \;=\; \sum_{n=0}^{\infty}(n+1)\,m_{n}^{2},
    \qquad m_{n} = \int r(y)^{n}\psi(y)\, d\rho(y).
\]
Every term of the sum is nonnegative, so $m_{n}=0$ for all $n \ge 0$. Since $r$ is a polynomial of degree one in $y$, the functions $r^{n}$ span the polynomials, so every polynomial moment of the finite signed measure $\psi \, d\rho$ vanishes.
Polynomials are uniformly dense in the continuous functions on the compact interval $[\eps, 1-\eps]$, so $\psi\, d\rho$ integrates every continuous function to zero and is therefore the zero measure. Hence $\psi(y) = 0$ for $\rho$-almost every $y$.

Finally, $1-p = \eps > 0$ and $1+p-y > 0$ for $y \in [\eps,1-\eps]$, so the only zero of $\psi$ is $b(p-y) = c$, that is $y = p - c/b = 1-c/b-\eps = q^{*}$, which lies in $[\eps,1-\eps]$ by the hypothesis $c/b < 1-2\eps$.
A probability measure concentrated on a single point is the Dirac mass there.
\end{proof}

There is a simple intuition for Theorem~2, based on the gradient alone.
The sign of each pairwise gradient in the $q$ direction is set by the opponent: $\psi(y)$ is positive when the opponent's degree of reciprocity, $r(y)$, exceeds $c/b$, and negative when it falls below $c/b$, the threshold described in Sec.~II\,B.
Two agents with different forgiveness weigh the same opponents differently, because the kernel $\varphi(q,y)$ increases with $r(y)$, and increases faster for a less forgiving focal agent.
A less forgiving agent therefore weights the positive contributions of highly reciprocal opponents more heavily than a more forgiving agent does.
If the more forgiving agent's gradient vanishes, the less forgiving agent's gradient is positive, unless every opponent contributes zero, that is, unless every agent is at $q^{*}$.
Hence no distribution supported at two or more values of $q$ satisfies the hypothesis of the theorem, and a distribution concentrated at a single value of $q$ satisfies it only if that value is a zero of $\psi$, that is, only at GTFT.
The proof above reaches the same conclusion by a different argument: it weights each agent's gradient by the agent's own numerator $\psi(q)$ and averages over the population, which turns the vanishing of every gradient into the
vanishing of a sum of squares.

Note that ``stationary'' is used loosely in Theorem~2: a distribution on $p = 1-\eps$ could be stationary under our projection system if all gradients are pointed perpendicular to the $p = 1-\eps$ line, i.e. purely in the $p$-direction.
Thus, we consider any distribution where the gradient in the $q$-direction is zero. (In other words, this is without the necessary restriction that the gradient in the $p$-direction be non-negative.)

The proof imposes no restriction on the support, so it covers a single agent, a finite population of $N$ agents, and a continuous density alike; the finite case is the one realized by the agent-based model.
Two remarks follow from \eqref{eq:qgrad} at no extra cost.
First, the hypothesis $c/b < 1-2\eps$ is sharp: if $c/b > 1-2\eps$ then $\psi < 0$ throughout $[\eps,1-\eps]$, every agent on the reciprocal line has a strictly negative gradient in the $q$ direction, and no stationary distribution exists there at all.
Second, the case of a population entirely at ALLC requires no separate argument, since $\psi(1-\eps) = -c\eps \neq 0$ and ALLC is therefore not a zero of $\psi$.

The strategy $(1-\eps, q^{*})$ satisfies $p - q = c/b$ exactly, and substituting $r_2 = c/b$ into the payoff gives $\pi(\mathbf{x}_1, (1-\eps,q^{*})) = b(1-\eps)-c$ for every strategy $\mathbf{x}_1$: an opponent at GTFT fixes the payoff of every co-player, whatever that co-player does.
Both components of the gradient against a population concentrated at GTFT therefore vanish identically, which is why the Dirac mass at GTFT is a fixed point of the full two-dimensional dynamics and not only of the motion along the line.

\section{Long-term experiment (Fig.~\ref{fig:long_run_trajectory})}

Directly evaluating the population gradient requires $N^2$ pairwise calculations at each learning step. Moreover, as the population approaches GTFT, the payoff gradients become small, so appreciable strategy changes require increasingly long timescales. To investigate the fate of the slow transient described in Sec.~III\,D of the main text, we investigated how a smaller population sampled from the full $N = 2500$ population at the end of its simulation would evolve. Importantly, this experiment was intended to estimate the qualitative long-term behavior, rather than reproduce the trajectory of the original population exactly.

We initialized the reduced population from the $N = 2500$ agent simulation at step 25000. At this point, all but one agent had $p = 1-\eps$. We retained this agent with $(p, q) \approx (0.678762, 0.999),$ as it is crucial to the population dynamics, and sampled 124 from those along the $p = 1-\eps$ boundary. To do this, we sorted the 2499 remaining agents by forgiveness $q$ and divided them into 124 contiguous groups: 105 with 20 agents, 19 with 21 agents. We selected one agent uniformly at random from each group, retaining its strategy and learning rate. This stratified sampling process ensures the $q$ distribution roughly resembles that of the full population.

Each sampled agent was assigned a fixed weight equal to the fraction of the original population that its stratum contained, i.e. a sampled agent from a group of size $m$ received weight $m/2500$. Importantly, this meant the agent with $p < 1-\eps$ received weight $1/2500.$ Without this weighting procedure, this agent would have twentyfold higher impact on the gradient of all agents, potentially impacting the population dynamics. Instead, this way, we can use a reduced population but still estimate the feedback dynamics between this agent and the 2499 others more precisely. The gradient for each agent is thus computed as
\[
G_i^t = \sum_{j=1}^{125} w_j \nabla_{x_i^t} \pi(x_i^t, x_j^t), \qquad \sum_{j=1}^{125}w_j = 1.
\]
To accelerate movement in the slow transient, we multiplied sampled learning rates by 100, but preserved their relative variation. We then performed $10^{10}$ additional updates:
\[
    x_i^{t+1} = \mathrm{Proj}_{[\eps, 1-\eps]^2} (x_i^t + 100\eta_i G_i^t).
\]

The calculation was implemented in compiled C using double-precision arithmetic. Reducing the population from 2500 to 125 agents alone decreases the pairwise computation by a factor of 400. We obtained a further reduction by expanding the interaction kernel in weighted population moments. Writing $r_j = p_j - q_j$, the pairwise gradients contain the factor $(1-r_ir_j)^{-2}$. For interaction partners with degrees of reciprocity within an interval of midpoint $a$ and half-width $h$, write $r_j = a+ hs_j$, where $|s_j|\leq 1$. Then
\[
    \frac{1}{(1-r_ir_j)^2} = \frac{1}{(1-ar_i)^2} \sum_{k=0}^{\infty} (k+1) \left(\frac{hr_i}{1-ar_i}\right)^k s_j^k.
\]
This expansion allows weighted moments of the opponent population to be computed once and reused for every agent. With $K$ retained terms, the calculation requires order $NK$ operations rather than $N^2$. At each update, the agent with the smallest $r_j$ was evaluated directly, and excluded from the expansion interval, narrowing the range occupied by the remaining population.

The number of terms was chosen adaptively so that the absolute truncation error in each gradient component was bounded by $10^{-14}$ before floating-point rounding. When too many terms were required, the calculation reverted to direct pairwise summation. Every learning step was still evaluated. The truncation tolerance controls the approximation of an individual gradient, not the accumulated error in the trajectory.

The reduced population moves progressively closer to GTFT, with the initially offset agent first increasing its reciprocity $p$ and then its forgiveness $q$ (Fig.~\ref{fig:long_run_trajectory}). After $10^{10}$ updates, the agent is at $(p, q) \approx (0.998994,0.502090)$, while the other agents have $p = 0.999,$ and $q$ values close to $0.499 = 1-c/b-\eps.$ The magnitudes of the payoff gradients also decrease significantly. These observations are consistent with a slow approach to GTFT, although the population has not reached that point exactly.

\section{No recovery to GTFT under uniform learning rates}
\label{sec:no-recovery-uniform}

We now consider a population with uniform learning rates in which every agent is displaced from GTFT to the same strategy.
If the displacement lowers the degree of reciprocity below $c/b$, the population cannot recover to GTFT, provided the learning rate satisfies the bound below.
This conclusion follows from the full discrete update rule, including projection, and does not require the displacement to be infinitesimal.

\begin{theorem}[No recovery under uniform learning rates]
\label{thm:no-recovery-uniform}
Let $0<\eps<1/2$, $b>c>0$, and write $\alpha=c/b$, with $\alpha<1-2\eps$.
Suppose all $N$ agents initially occupy the same strategy $x^0=(p^0,q^0)\in[\eps,1-\eps]^2$, with
\[
    r^0=p^0-q^0<\alpha,
\]
and share a learning rate $\eta>0$ satisfying
\[
    \eta b\leq(1-\alpha)^2.
\]
Then all agents remain at a common strategy $x^t=(p^t,q^t)$, with $r^t<\alpha$ at every learning step. Both $p^t$ and $q^t$ are nonincreasing, and $p^t<1-\eps$ for every $t\geq 1$.
In particular, the population cannot return to GTFT.
\end{theorem}

\begin{proof}
Agents with identical strategies and learning rates receive identical updates.
Consequently, the population remains concentrated at a single strategy $x=(p,q)$, whose update is
\[
    x^{t+1}=\operatorname{Proj}_{[\eps,1-\eps]^2} \left(x^t+\eta G(x^t)\right),
    \qquad
    G(x) = \left.\nabla_{x_1}\pi(x_1,x_2)\right|_{x_1=x_2=x}.
\]
Substituting $r=p-q$ into the pairwise gradient gives
\[
    G_p(x)=\frac{bq(r-\alpha)}{(1-r)^2(1+r)},
    \qquad
    G_q(x)=\frac{b(1-p)(r-\alpha)}{(1-r)^2(1+r)}.
\]
Since $q\geq\eps$, $1-p\geq\eps$, and $|r|\leq 1-2\eps$, both components are strictly negative whenever $r<\alpha$.
Thus, as long as this inequality holds, neither coordinate increases under the projected update.

It remains to show that $r<\alpha$ is preserved. Let
\[
    \widetilde p=p+\eta G_p(x),
    \qquad
    \widetilde q=q+\eta G_q(x)
\]
denote the coordinates before projection, and set $\widetilde r=\widetilde p-\widetilde q$.
Subtracting the two gradient components yields
\[
    \widetilde r-\alpha
    =
    (r-\alpha)
    \left[
        1-\frac{\eta b(1-p-q)}{(1-r)^2(1+r)}
    \right].
\]
If $1-p-q\leq 0$, the bracket is at least one.
Otherwise, $0<1-p-q<1+r$, and $r<\alpha$ implies
\[
    0
    <
    \frac{\eta b(1-p-q)}{(1-r)^2(1+r)}
    <
    \frac{\eta b}{(1-r)^2}
    \leq
    \frac{\eta b}{(1-\alpha)^2}
    \leq 1.
\]
The bracket is therefore strictly positive in either case, so
$\widetilde r<\alpha$.

Projection also preserves the inequality.
Neither coordinate can cross the upper boundary because both gradients are negative.
If $\widetilde p\leq\eps$, projection gives $p^{t+1}=\eps$ and hence $r^{t+1}\leq 0<\alpha$.
Otherwise, projecting $\widetilde q$ upward to $\eps$, if necessary, can only decrease the degree of reciprocity.
Thus $r^{t+1}<\alpha$, and induction establishes the claim at every step.

Finally, $G_p<0$ whenever $r<\alpha$, so reciprocity decreases strictly whenever $p>\eps$.
If the population starts on $p=1-\eps$, it leaves that boundary at the first update.
Since reciprocity never subsequently increases, the population cannot return to the reciprocal boundary or to GTFT.
\end{proof}

There is a simple intuition for Theorem~\ref{thm:no-recovery-uniform}, based on the gradient alone. Both components of each pairwise gradient contain the factor $b r_2 - c$: raising either cooperation probability increases an agent's payoff only when its co-player's degree of reciprocity exceeds $c/b$ (Sec.~II\,B of the main text). In a population concentrated at a single strategy, every co-player has the agent's own degree of reciprocity, $r$, so when $r < c/b$ both components of the gradient are negative, and every agent lowers both its reciprocity and its forgiveness. A return to GTFT would require reciprocity to increase, which the gradient permits only when $r$ exceeds $c/b$. The degree of reciprocity can increase, because forgiveness may fall faster than reciprocity, but each learning step changes $r$ by an amount proportional to its distance from $c/b$, and the bound $\eta b \le (1-\alpha)^2$ ensures that no step changes $r$ by more than that distance. Hence $r$ remains below $c/b$, and the population never returns to the reciprocal boundary. With variable learning rates, the premise of the theorem fails at the first learning step: agents with different learning rates move different distances, so the population no longer occupies a single strategy, and the slow learners retain high reciprocity (Fig.~\ref{fig:gtft-smallq-perturbation}).

For our parameters $b=1$, $c=1/2$, and $\eta=0.2$, the learning-rate condition is satisfied because $0.2\leq(1-1/2)^2=0.25$.
In particular, displacing every agent from GTFT to
\[
    (p^0,q^0)
    =
    \left(1-\eps,\;1-c/b-\eps+\delta\right),
    \qquad
    0<\delta\leq c/b,
\]
gives $r^0=c/b-\delta<c/b$, so Theorem~\ref{thm:no-recovery-uniform} excludes recovery following any such common increase in forgiveness.
The result applies to a population concentrated at one displaced strategy (as in Fig.~\ref{fig:gtft-smallq-perturbation}); it does not establish the same conclusion when different agents receive different displacements (as in Fig.~\ref{fig:gtft-perturbation-population}, though these displacements also do not always lower the degree of reciprocity).

\newpage

\section*{Supplemental Figures}

\setcounter{figure}{0}
\renewcommand{\thefigure}{S\arabic{figure}}
\renewcommand{\theHfigure}{S\arabic{figure}}

\begin{figure}[h]
\centering
\includegraphics[width=\textwidth]{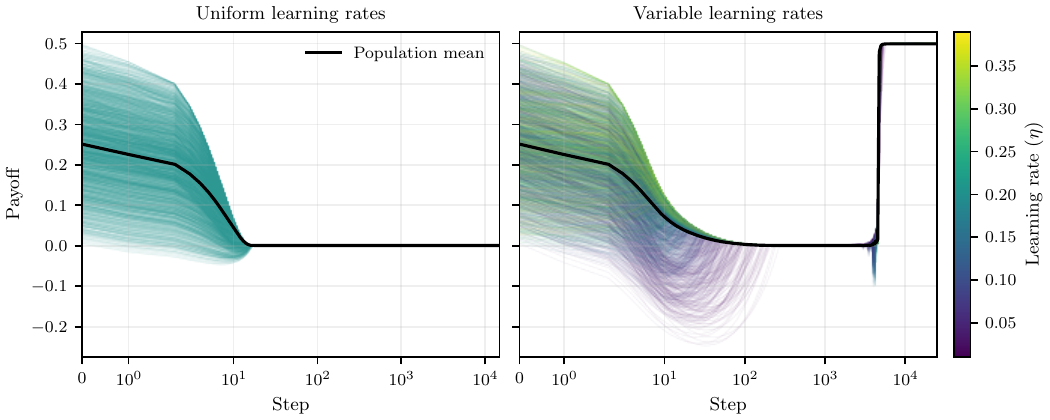}
\caption{\textbf{Populations with variable learning rates can recover socially optimal strategies.} Payoffs for agents with uniform learning rates (left) and variable learning rates (right) over time. Values are evaluated from the same simulation as Figs.~2 and~3 of the main text.}
\label{fig:avg-payoff}
\end{figure}

\begin{figure}[h]
\centering
\includegraphics[width=\textwidth]{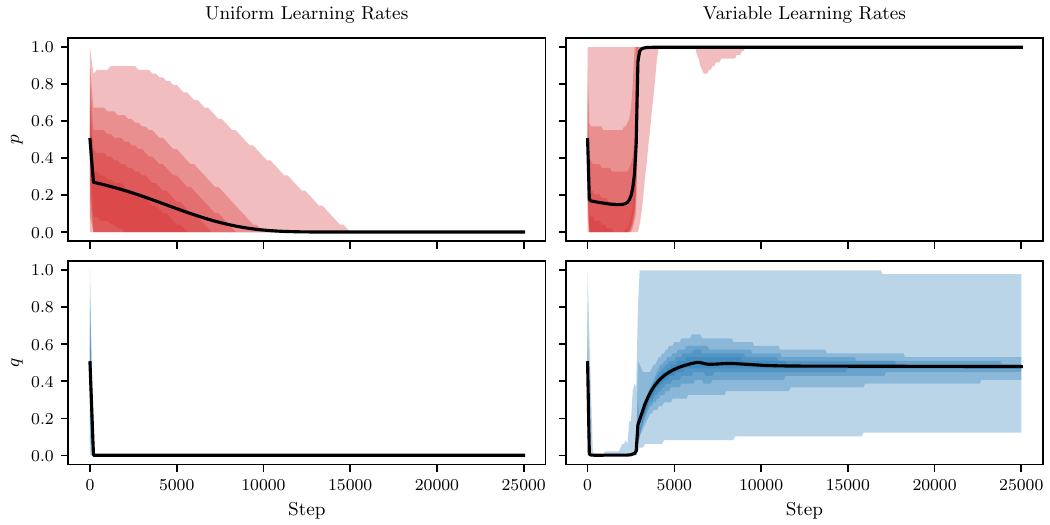}
\caption{\textbf{The infinite-population model reproduces the same divergence of outcomes as the agent-based model.}
Distributions of $p$ (top) and $q$ (bottom) over time in the infinite-population model (Sec.~\ref{sec:infinite-population}), with uniform learning rates (left) and variable learning rates (right).
The population is represented by 100 grids, one per learning rate, each of $50 \times 50$ points evenly spaced on $[\varepsilon, 1-\varepsilon]^2$; every grid carries total frequency $1/100$.
With variable learning rates, the grids' rates are evenly spaced from $\eta = 0.01$ to $\eta = 0.39$; with uniform learning rates, all are $\eta = 0.2$.
Initially, every grid point has equal frequency.
Nested shaded bands contain the central $99.9\%$, $80\%$, $60\%$, $40\%$ and $20\%$ of the frequency-weighted $p$ or $q$ values, darker toward the center; the black line is the mean.}
\label{fig:grid-distributions}
\end{figure}

\begin{figure}[h]
\centering
\includegraphics[width=\textwidth]{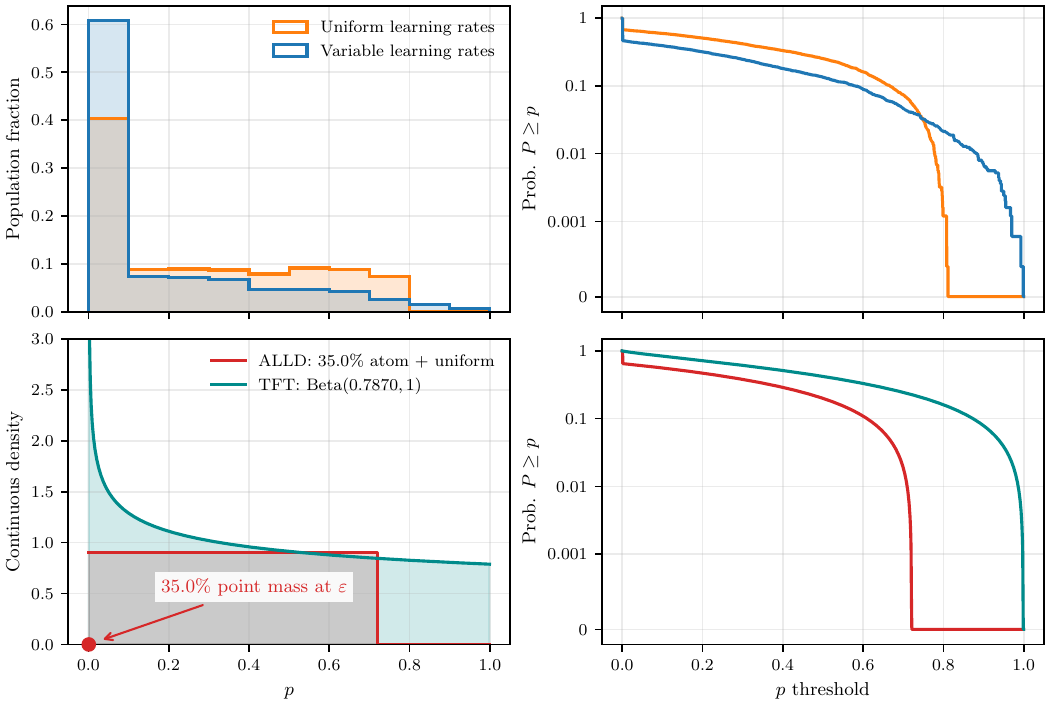}
\caption{\textbf{The distribution of $p$ values upon reaching $q = \eps$ has support on higher $p$ when agents have variable learning rates.} Top row: histograms (left) and survival functions (right) of the $p$ values of agents on the first timestep where every agent has $q = \eps$. Uniform learning rates produce a distribution with a higher mean, but lower maximum. Bottom row: densities (left) and survival functions (right) of distributions that satisfy the conditions of Theorem~1 of the main text. ALLD distribution: atom with frequency 0.35 at ALLD, rest of population uniformly distributed from $\eps$ to approximately $0.72$. TFT distribution: $\mathrm{Beta}(0.7870, 1)$ scaled onto $[\eps, 1-\eps]$.}
\label{fig:tus}
\end{figure}

\begin{figure}[h]
\centering
\includegraphics[width=\textwidth]{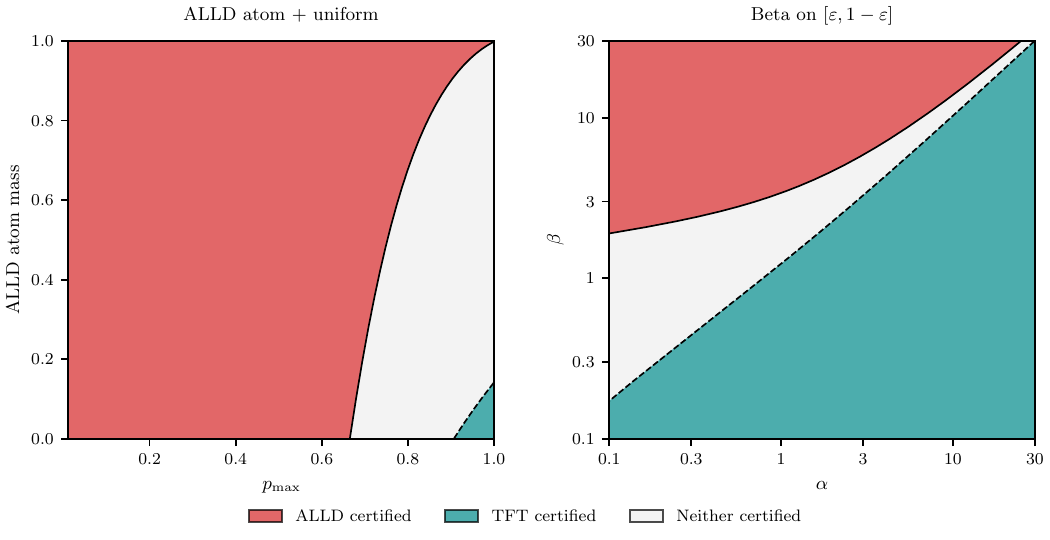}
\caption{\textbf{Theorem~1 can determine whether a wide range of distributions of $p$ on $q = \eps$ will converge to ALLD or TFT.} Left: distributions with a point mass at ALLD ($p = \eps$) and the rest of the population uniformly distributed from $\eps$ to $p_{\max}$. Right: Beta distributions, scaled onto the region $[\eps, 1-\eps]$. Red denotes a parameter regime where the starting distribution converges to ALLD in finite steps; teal a parameter regime where the distribution converges to TFT in finite steps; gray where the outcome cannot be determined by Theorem~1 alone. The dashed line corresponds to $P(\mu) - Q(\mu) = 0$, and the solid black line to $U(\mu) = 0$.}
\label{fig:tus-phase-diagram}
\end{figure}

\begin{figure*}[h]
\centering
\includegraphics[width=\textwidth]{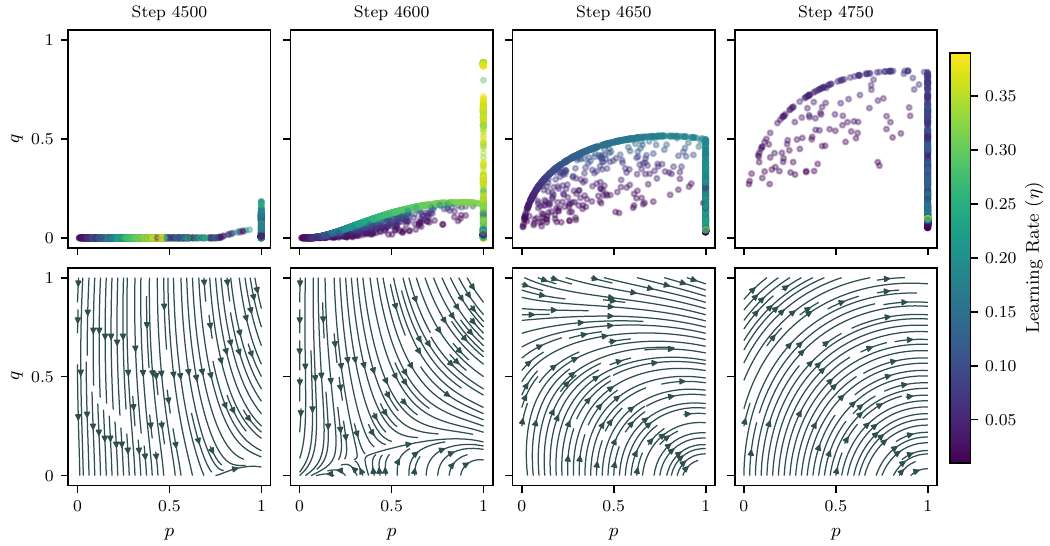}
\caption{\textbf{Agents leave the unforgiving line in order of learning rate, and the slower they learn the higher the forgiveness they reach.} Agents' strategies (top, colored by learning rate $\eta$) and streamline plots of the population gradient field (bottom) at four learning steps. \emph{Step 4500}: agents occupy the unforgiving line $q = \eps$, with a buildup of high-reciprocity agents near TFT. \emph{Step 4600}: agents at TFT move up the $p = 1-\eps$ boundary while faster learners arc through the interior. \emph{Steps 4650} and \emph{4750}: slower learners follow wider arcs toward higher $q$ on the $p = 1-\eps$ line, producing the initial GTFT-like distribution.
}
\label{fig:move-to-GTFT}
\end{figure*}

\begin{figure*}[h]
\centering
\includegraphics[width=\textwidth]{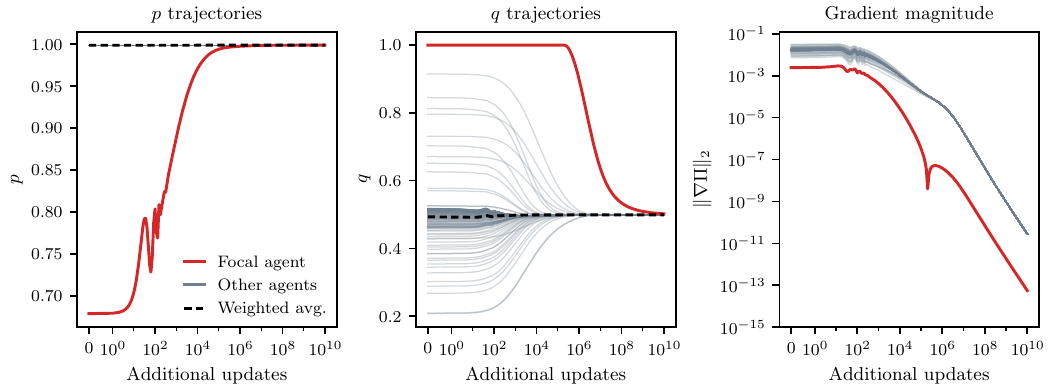}
\caption{\textbf{Estimate of the variable learning rate population's long-term trajectory predicts slow convergence to GTFT.} Starting from the $N=2500$ agent population at step 25000, we constructed a weighted sample of 125 agents. The sole agent with $p < 1-\eps$ was retained, and one agent was sampled uniformly from 124 nearly equal-sized strata of the remaining population, sorted by $q$. Each representative of a stratum was given weight according to the proportion of agents within that stratum relative to the whole population. The agent with $p < 1-\eps$ retained weight 1/2500 (without this weighting procedure, this agent's weight would be 1/125). Learning rates were multiplied by 100 for computational efficiency; this is not meant to simulate exactly what would happen in a timescale 100 times longer. The sampled population then underwent $10^{10}$ updates. Left and center: individual $p$ and $q$ trajectories, with the initially aberrant agent colored in red, and population mean values in dashed black lines. Right: the magnitude of the gradient (before scaling by learning rates) for each agent.}
\label{fig:long_run_trajectory}
\end{figure*}

\clearpage

\begin{figure*}[h]
\centering
\includegraphics[width=\textwidth]{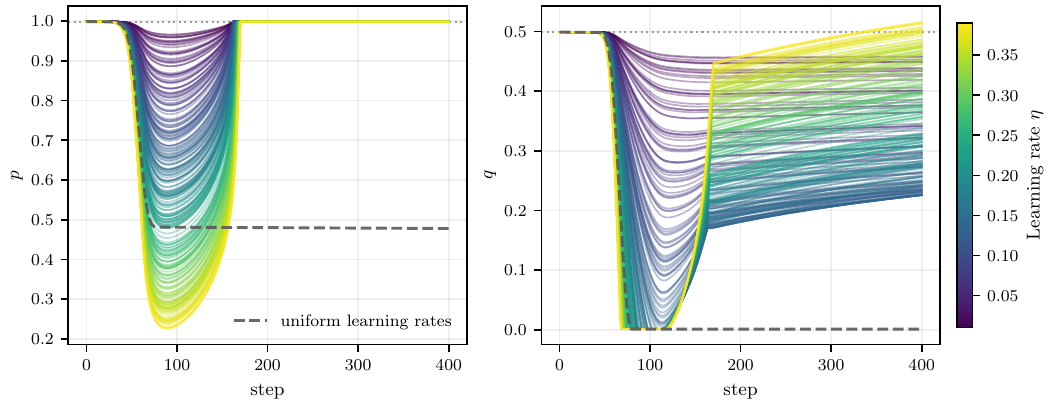}
\caption{\textbf{Recovery trajectory from a small $q$ perturbation for a GTFT population with variable learning rates.} All $N = 250$ agents received a perturbation $q\mapsto q+10^{-6}$ from GTFT. This perturbation was tested with uniform learning rates $\eta_i = 0.2$ (dashed line) and variable learning rates $\eta_i \sim \mathrm{Uniform}(0.01, 0.39)$ (colored lines). Note that the uniform learning population, when initialized from the same point, will thereafter share a location on the strategy space.}
\label{fig:gtft-smallq-perturbation}
\end{figure*}

\begin{figure*}[h]
\centering
\includegraphics[width=\textwidth]{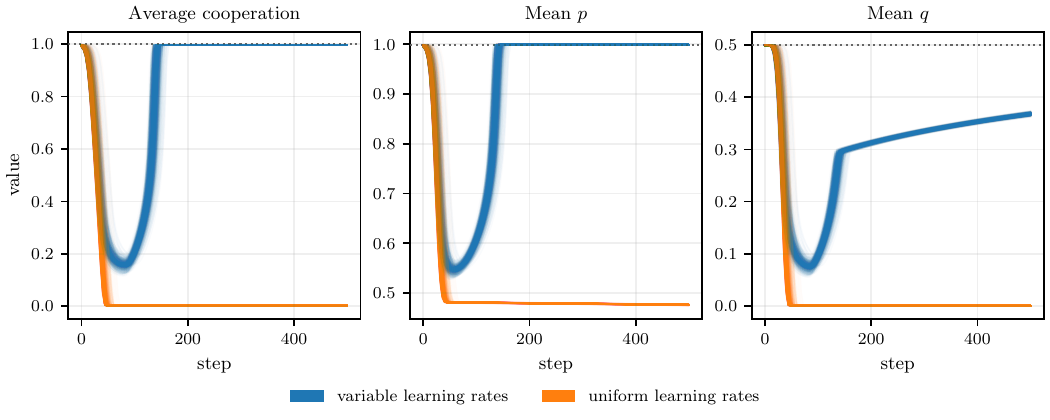}
\caption{\textbf{Populations with variable learning rates can recover from shocks to a GTFT distribution, but those with uniform learning rates can't.} A fraction $\mathrm{round}(xN)$, of $N = 250$ agents, where $x \sim \mathrm{Uniform(0.01, 1.0)}$, are perturbed area-uniformly onto the half-disk of radius 0.01 around GTFT in the strategy space, with each displacement drawn independently. Each perturbation regime was tested with uniform learning rates $\eta_i = 0.2$ or variable learning rates $\eta_i \sim \mathrm{Uniform}(0.01, 0.39).$ Results were collected over $M=500$ tested populations. The average cooperation (left), mean $p$ (middle), and mean $q$ (right) are shown for every one of $M$ populations.}
\label{fig:gtft-perturbation-population}
\end{figure*}

\begin{figure*}[h]
\centering
\includegraphics[width=\textwidth]{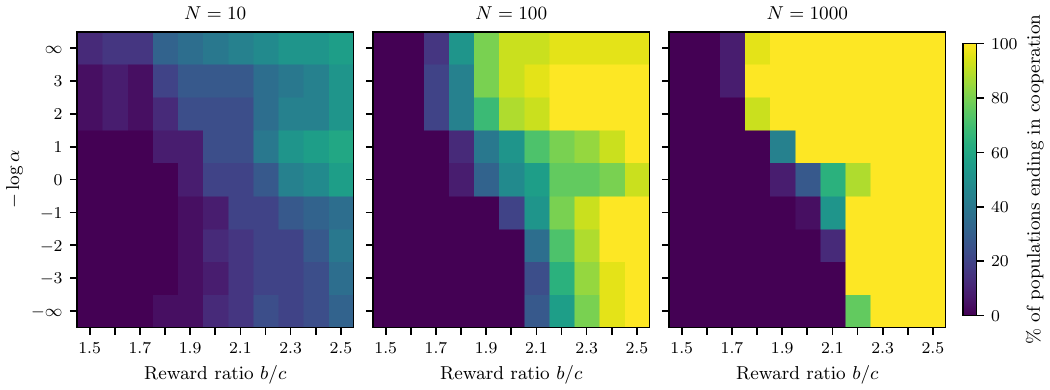}
\caption{\textbf{Cooperation is favored by a larger benefit-to-cost ratio, greater heterogeneity in learning rates, and larger populations.}
Each panel shows the percentage of 25 replicate populations ending in cooperation (color) as a function of the benefit-to-cost ratio $b/c$ (horizontal axis) and the heterogeneity of learning rates, $-\log_{10}\alpha$ (vertical axis), for populations of size $N = 10$ (left), $N = 100$ (center) and $N = 1000$ (right).
Learning rates are drawn from a symmetric $\mathrm{Beta}(\alpha, \alpha)$ distribution scaled onto $[0.01, 0.39]$. Its variance increases with $-\log_{10}\alpha$: at $-\log_{10}\alpha = -\infty$, every rate is $0.2$; at $0$, the rates are uniform on $[0.01, 0.39]$; at $+\infty$, half the agents have rate $0.01$ and half $0.39$.
Initial strategies are drawn uniformly from $[\varepsilon, 1-\varepsilon]^2$, and the same 25 initial populations (and learning-rate quantiles) are used in every cell of a panel.
Every 100 steps, a population is classified as cooperative if its mean $p$ exceeds $0.99$ or its average cooperation exceeds $0.95$, and as defecting if every agent has $p < 0.5$ or the fraction of agents at ALLD exceeds $(b/c - 1)/(b/c - 1/2)$; populations still unclassified after 50{,}000 steps count as cooperative if their average cooperation exceeds $0.5$.
The transition between defection and cooperation sharpens with $N$.}
\label{fig:popsize+b/c comparison}
\end{figure*}

\begin{figure*}[h]
\centering
\includegraphics[width=\textwidth]{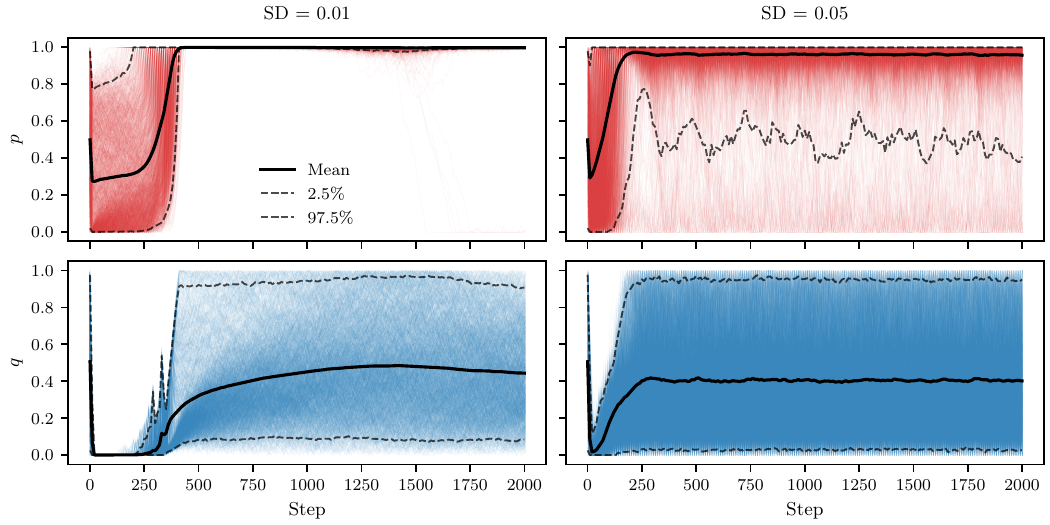}
\caption{\textbf{Noisy gradient ascent with uniform learning rates also produces a GTFT-like population.} Individual agents' $p$ (top, red) and $q$ (bottom, blue) values over learning steps, with low multivariate Gaussian noise (left; SD = 0.01) and high noise (right; SD = 0.05) added to each strategy update. The black line shows the mean and dashed lines contain 95\% of agents within them. In both cases reciprocity $p$ rises to near $1$ and forgiveness $q$ settles near the GTFT value, with a wider stationary spread under high noise.}
\label{fig:noisy-distributions}
\end{figure*}

\begin{figure*}[h]
\centering
\includegraphics[width=\textwidth]{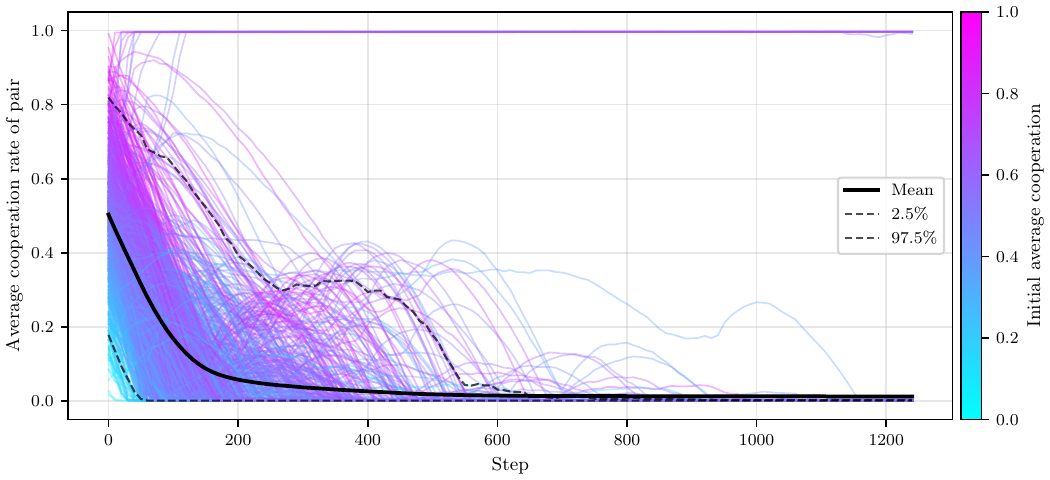}
\caption{\textbf{An isolated pair of noisy gradient learners fails to cooperate.}
Average cooperation rate of 1000 independent pairs of agents over learning steps.
In each pair, the two agents play only each other (no self-play), starting from strategies drawn uniformly from $[\varepsilon, 1-\varepsilon]^2$, and update by gradient ascent with learning rate $\eta = 0.02$ plus isotropic Gaussian noise (SD $= 0.001$) added to each update.
Each line is one pair, colored by its initial average cooperation rate; the black line shows the mean across pairs and dashed lines the 2.5\% and 97.5\% quantiles.
By step 1240, $98.9\%$ of pairs have fallen below $5\%$ cooperation; the remaining $1.1\%$, with initial cooperation between $0.43$ and $0.87$, remain near full cooperation.}
\label{fig:noisy-cooperation}
\end{figure*}

\end{document}